\documentclass[a4paper,UKenglish]{lipics-v2016}
\usepackage[utf8]{inputenc}

\usepackage{wrapfig}
\usepackage{amssymb}
\usepackage[utf8]{inputenc}
\usepackage{amsthm}

\usepackage{thmtools}
\usepackage{thm-restate}

\usepackage{algorithm}
\usepackage[noend]{algpseudocode}
\usepackage{bm}
\usepackage{array}
\usepackage{tikz}
\usepackage{pgf,tikz}
\usepackage{mathrsfs}
\usetikzlibrary{arrows}
\usepackage{pgfplots}
\usepackage{booktabs,makecell}
\usepackage{rotating}
\usepackage{hyperref}
\usepackage{bbm}
\usepackage{enumerate}
\usepackage{tikz}
\usepackage{pgfplots}
\usepackage{url}

\theoremstyle{definition}
\newtheorem{thm}{Theorem}[section]

\newtheorem{lem}[thm]{Lemma}
\newtheorem{proposition}[thm]{Proposition}
\newtheorem{defi}[thm]{Definition}

\renewcommand{\deg}{\operatorname{deg}}
\newcommand{\dist}{\operatorname{dist}}
\newcommand{\disc}{\operatorname{disc}}
\newcommand{\M}{\mathbf{M}}
\renewcommand{\P}{\mathbf{P}}
\newcommand{\Ex}[1]{\mathbb{E}\left[{#1}\right]}
\newcommand{\Pro}[1]{\mathbb{P}\left[{#1}\right]}
\newcommand{\Psub}[2]{\mathbb{P}_{#1}\left[{#2}\right]}
\newcommand{\Esub}[2]{\mathbb{E}_{#1}\left[{#2}\right]}
\newcommand{\Var}{\mathrm{Var}}
\newcommand{\Uni}{\mathsf{Uni}}
\newcommand{\Poi}{\mathsf{Poi}}
\newcommand{\Bin}{\mathsf{Bin}}
\newcommand{\Geo}{\mathsf{Geo}}

\usepackage{xspace,setspace,color}   
\newcounter{nummer}

\numberwithin{equation}{section}

\title{Randomized Load Balancing on Networks with Stochastic Inputs}

\author[1]{Leran Cai}
\author[1]{Thomas Sauerwald}
\affil[1]{University of Cambridge, email: \texttt{(lc647|tms41)@cl.cam.ac.uk}}

\authorrunning{Leran Cai and Thomas Sauerwald} 

\Copyright{Leran Cai and Thomas Sauerwald}


\numberwithin{thm}{section}
\begin{document}

\maketitle
%
%
%

\begin{abstract}
	Iterative load balancing algorithms for indivisible tokens have been studied intensively in the past, e.g.,~\cite{RSW98,MGS98,SS12}. Complementing previous worst-case analyses, we study an average-case scenario where the load inputs are drawn from a fixed probability distribution. For cycles, tori, hypercubes and expanders, we obtain almost matching upper and lower bounds on the discrepancy, the difference between the maximum and the minimum load. Our bounds hold for a variety of probability distributions including the uniform and binomial distribution but also distributions with unbounded range such as the Poisson and geometric distribution. For graphs with slow convergence like cycles and tori, our results demonstrate a substantial difference between the convergence in the worst- and average-case.
	An important ingredient in our analysis is new upper bound on the $t$-step transition probability of a general Markov chain, which is derived by invoking the evolving set process.
\end{abstract}

%

\section{Introduction}

In the last decade,
large parallel networks became widely available
for industrial and academic users.  An important prerequisite
for their efficient usage is to balance their work efficiently.
Load balancing is known to have applications
to 
scheduling~\cite{Surana06},
routing~\cite{Cyb89},
numerical computation such as solving partial differential equations~\cite{Zhanga09,Williams91,SubramanianScherson94},
and finite element computations~\cite{FEM}.
In the standard abstract formulation of load balancing, processors are represented
by nodes of a graph, while links are represented by edges.
The objective is to balance the load by allowing nodes to exchange loads with their neighbors via the incident edges. In this work we will study a decentralized and iterative load balancing protocol where a processor knows only its current load and that of the neighboring processors and based on this, decides how much load should be sent (or received). 

{\bf Load Balancing Models.} 
A widely used approach is diffusion, e.g., the first-order-diffusion scheme \cite{Cyb89,MGS98}, where the amount of load sent along each edge in each round is proportional to the load difference between the incident nodes. In this work, we consider the alternative, the so-called matching model, where in each round only the edges of the matching are used to average the load locally. In comparison to diffusion, the matching model reduces the communication in the network and moreover tends to behave in a more ``monotone'' fashion than diffusion, since it avoids concurrent load exchanges which may increase the maximum load or decrease the minimum load in certain cases.

We measure the smoothness of the load distribution by the so-called {\em discrepancy} which is the difference between the maximum and minimum load among all nodes. In view of more complex scenarios where jobs are eventually removed or new jobs are generated, the discrepancy seems to be a more appropriate measure than the {\em makespan}, which only considers the maximum load.

Many studies in load balancing assume that load is arbitrarily divisible. In this so-called {\em continuous case}, load balancing corresponds to a Markov chain on the graph and one can resort to a wide range of established techniques to analyze the convergence speed \cite{Boillat90,GLM99,MGS98}. In particular, the {\em spectral gap} captures the time to reach a small discrepancy fairly accurately, e.g., see \cite{SJ89,RSW98} for the diffusion and see \cite{BGPS06,MG96} for the matching model.

However, in many applications a processor's load may consist of tasks which are not further divisible, which is why the continuous case has been also referred to as ``idealized case''~\cite{RSW98}. A natural way to model indivisible tasks is the {\em unit-size token model} where one assumes a smallest load entity, the unit-size token, and load is always represented by a multiple of this smallest entity. In the following, we will refer to the unit-size token model as the {\em discrete case}. 

Initiated by the work of \cite{RSW98}, there has been a number of studies on load balancing in the discrete case. Unlike \cite{RSW98}, \cite{SS12} analyzed a randomized rounding based strategy, meaning that an excess token will be distributed uniformly at random among the two communicating nodes. The authors of \cite{SS12} proved that with this strategy the time to reach constant discrepancy in the discrete case is essentially the same as the corresponding time in the continuous case. Their results hold both for the {\em random matching model}, where in each round a new random matching is generated by a simple distributed protocol, and the {\em balancing circuit model} (a.k.a. dimension exchange), where a fixed sequence of matching is applied periodically. In this work, we will focus on the {\em balancing circuit model}, which is particularly well suited for highly structured graphs such as cycles, tori or hypercubes.

\textbf{Worst-Case vs. Average-Case Inputs.}
Previous work has almost always adopted the usual worst-case framework for deriving bounds on the load discrepancy \cite{RSW98}.  That means that any upper bound on the discrepancy holds for an arbitrary input, i.e., an arbitrary initial load vector. While it is of course very natural and desirable to have such general bounds, the downside is that for graphs with poor expansion like cycles or 2D-tori, the convergence is rather slow, i.e., quadratic or linear in the number of nodes $n$.

This serves as a motivation to explore an average-case input. Specifically, we assume that the number of load items at each node is sampled independently from a fixed distribution. Our main results demonstrate that the convergence of the load vector is considerably quicker (measured by the load discrepancy), especially on networks with slow convergence in the worst-case such as cycles and 2D-tori.

We point out that many related problems including scheduling on parallel machines or load balancing in a dynamic setting (meaning that jobs are continuously added and processed) have been studied under random inputs, e.g.,~\cite{AKU05,GI99,ACS16}. To the best of our knowledge, only very few works have studied this question in iterative load balancing. One exception is \cite{S99}, which investigated the performance of continuous load balancing on tori in the diffusion model. In contrast to this work, however, only upper bounds are given and they hold for the multiplicative ratio between maximum and minimum load, rather than the discrepancy.



Our main results in this paper hold for all distributions satisfying the following definition, which is satisfied by the uniform, binomial, Poisson and geometric distribution (see Section~\ref{sec:def}).
\newpage
\begin{defi}\label{def:expcon}
We say that a distribution $D$ over $\mathbb{N} \cup \{0\}$ is exponentially concentrated if there is a constant $ \kappa >0$ so that for any $X \sim D$, $\delta > 0$,
\begin{align*}
  \Pro{  | X - \mu | \geq \delta \cdot \sigma} \leq \exp\left( -\kappa \delta \right),
\end{align*}
where $\mu$ and $\sigma^2$ are the expectation and variance of $D$. In the following, we refer to \textbf{average-case} when the initial number of load items on each vertex is drawn independently from a fixed exponentially concentrated distribution.
\end{defi}

\textbf{Our Results.} 
Our first contribution is a general formula that allows us to express the load difference between an arbitrary pair of nodes in round $t$. 
Here the round matrix $\M$ is the product of the matching matrices that are applied periodically (cf.~Section \ref{sec:def}). 

\begin{restatable}{thm}{mainone}\label{thm:main1}
Consider the balancing circuit model with an arbitrary round matrix $\M$ in the average case. Then for any pair of nodes $u,v$ and round $t$, it holds for any $\delta > 0$ that
\begin{equation*}
	\Pro{\left| x_u^{(t)} - x_{v}^{(t)} \right| \geq \delta \cdot \sqrt{128} \kappa \cdot \sigma \cdot \log n \cdot 
		\left\| \M^t_{.,u} - \M^t_{.,v} \right\|_2 + \sqrt{48 \log n}} \leq 2 \cdot e^{-\delta^2} + 2 n^{-2}.
	\end{equation*}
Further, for any pair of vertices $u,v$ and any round $t$ satisfying $t=\omega(1)$,

	\label{thm:lowerboundforcycles}
	\begin{equation*}
	\Pro{\left| x_u^{(t)} - x_{v}^{(t)} \right| \geq \sigma/(2 \sqrt{2\log_2 \sigma}) \cdot \left\| \M^t_{.,u} - \M^t_{.,v} \right\|_2 - \sqrt{48 \log n} } \geq \frac{1}{16}. 
	\end{equation*}
\end{restatable}

The proof of the upper bound Theorem~\ref{thm:main1} is the easier direction, and it relies on a previous result relating continuous and discrete load balancing from \cite{SS12}. The lower bound is technically more challenging and applies a generalized version of the central limit theorem. 


Together, the upper and lower bound in the above result establish that the load deviation between any two nodes $u$ and $v$ is essentially captured by $ \left\| \M^t_{.,u} - \M^t_{.,v} \right\|_2$. However, in some instances it might be desirable to have a more tangible estimate at the expense of generality. A first step towards this goal is to observe that $ \left\| \M^t_{.,u} - \M^t_{.,v} \right\|_2^2 
\leq 4 \cdot \max_{k \in V} \| \M^t_{.,k} - \mathbf{\frac{1}{n}} \|_2^2$ (see~Lemma~\ref{lem:twonodestoavg}). Hence we are left with the problem of understanding the $t$-step probablity vector $\M^t_{.,k}$.



For reversible Markov chains, the last expression has been analyzed in several works, e.g., a result from~\cite[Lemma~3.6]{Lyons05} implies that for random walks on graphs, $\P_{u,v}^t=O(\deg(v)/\sqrt{t})$ (cf.~\cite{Lyons05}). However, the Markov chain associated to $\M$ is not reversible in general. For irreversible Markov chains, \cite{levin2009markov} use the so-called evolving set process to derive a similar bound. Specifically, they proved in \cite[Theorem~17.17]{levin2009markov} that if $\P$ denotes the transition matrix of a lazy random walk (i.e., a random walk with loop probability at least $1/2$) on a graph with maximal degree $\Delta$, then for any vertex $x \in V$: \[
  \left| \P^t_{x,x} - \pi_x \right| \leq \frac{\sqrt{2} \Delta^{5/2}}{\sqrt{t}},
\]
where $\pi$ is the stationary distribution of $\P$.
Such estimates have been used in various applications besides load balancing, including distributed random walks and spanning tree enumeration \cite{ANPT13,Lyons05}. Here we generalize this result to Markov chains with an arbitrary loop probability and to arbitrary $t$-step transition probabilities:

\begin{restatable}{thm}{markov}\label{thm:markov}
	Let $\mathbf{P}$ be the transition matrix of an irreducible Markov chain  and $\pi$ its stationary distribution. Then we have for all states $x, y$ and step $t$,
	\begin{equation*}
	\left|\mathbf{P}^t_{x,y} - \pi_y \right| \leq  \frac{\pi_{\max}^{3/2}}{\pi_{\min}^{3/2}}\cdot \frac{2}{\beta^{1/2}\alpha} \sqrt{\frac{1-\beta+\alpha}{\alpha t}},
	\end{equation*}
	where $\alpha := \min\limits_{u \neq v} \P_{u,v} > 0$ and $\beta := \min\limits_{u} \P_{u,u} > 0$.
\end{restatable}

Applying this bound to a round matrix $\M$ that is formed of $d=O(1)$ matchings, we obtain
$
	  \left| \M^t_{u,v} - 1/n \right| = O(t^{-1/2}). 
$
It should be noted that \cite[Lemma 2.5]{SS12} proved a weaker version where the upper bound is only $O(t^{-1/8})$ instead of $O(t^{-1/2})$. As we will prove in Lemma~\ref{lem:intuition}, the bound $O(t^{-1/2})$ is asymptotically tight if we consider the balancing circuit model on cycles.

Combining the bound in Theorem~\ref{thm:markov} with the upper bound in Theorem~\ref{thm:main1} yields:

\begin{restatable}{thm}{maintwo}\label{thm:maintwo}
Consider the balancing circuit model with an arbitrary round matrix $\M$ consisting of $d=O(1)$ matchings in the average case. Then the discrepancy after $t$ rounds is $O(t^{-1/4} \cdot \sigma \cdot (\log n)^{3/2} + \sqrt{\log n})$ with probability $1 - O(n^{-1})$.
\end{restatable}

Since the initial discrepancy in the average case is $O(\sigma \cdot \log n)$ (see Lemma~\ref{lem:distribution}), Theorem~\ref{thm:maintwo} implies that in the average case, there is a signficant decrease (roughly of order $t^{-1/4}$) in the discrepancy, regardless of the underlying topology.


For round matrices $\M$ with small second largest eigenvalue, the next result provides a significant improvement:
\begin{restatable}{thm}{mainthree}
\label{thm:main3}
Consider the balancing circuit model with an arbitrary round matrix $\M$ consisting of $d$ matchings in the average case. Then the discrepancy after $t$ rounds is $O(\lambda(\M)^{t/4} \cdot \sigma  \cdot (\log n)^{3/2}+\sqrt{\log n})$ with probability $1 - O(n^{-1})$.
\end{restatable}
Hence for graphs where $\lambda$ is bounded away from $1$, we even obtain an exponential convergence.

\begin{wrapfigure}{r}{4cm}
\begin{center}
\begin{tabular}{|c|c|}
\hline
Graph & $\disc(x^{(t)}) $ \\
\hline
Cycle & $t^{-1/4} \cdot \sigma$ \\
$r$-dim. Torus & $t^{-r/4} \cdot \sigma $ \\ 
Expander & $\lambda^{t/4} \cdot \sigma$ \\
Hypercube & $2^{-t/2} \cdot \sigma$ \\
\hline
\end{tabular}\label{fig:summary}
\caption{Discrepancy bounds (without logarithmic factors) for different topologies.}
\end{center} 
\end{wrapfigure}

%


In Section~\ref{sec:application}, we  derive bounds on the discrepancy for cycles, $r$-dim. Torus, expanders and hypercubes. A summary of these results can be found in Figure~1.

Finally, we discuss our results and contrast them to the convergence of the discrepancy in the worst-case in Section \ref{sec:discussion}. On a high level, these results demonstrate that on all the considered topologies, we have much faster convergence in the average-case than in the worst-case. However, if we are only interested in the time to achieve a very small, say, constant or poly-logarithmic discrepancy, then we reveal an interesting dichotomy: we have a quicker convergence than in the worst-case if and only if the standard deviation $\sigma$ is smaller than some threshold, which depends on the actual toplogy. We observe the same phenomena in our experiments, which are also discussed in Section~\ref{sec:experiments}.


\section{Notation and Background}\label{sec:def}

We assume that $G = (V,E)$ is an undirected, connected graph with $n$ nodes labelled in $[0,n-1]$. Unless stated otherwise, all logarithms are to the base $e$. The notations $\Pro{\mathcal{E}}$ and $\Ex{X}$ denote the probability of an event $\mathcal{E}$ and the expectation of a random variable $X$, respectively. For any $n$-dimensional vector $x$, $\disc(x)=\max_i x_i - \min_i x_i$ denotes the {\em discrepancy}.

\textbf{Matching Model.}
In the {\em matching model} (sometimes also called {\em dimension exchange model}),  every two matched nodes in round $t$ balance their load as evenly as possible. This can be expressed by a symmetric $n$ by $n$ matching matrix $\M^{(t)}$, where with slight abuse of notation we use the same symbol for the matching and the corresponding matching matrix. Formally, matrix $\M^{(t)}$ is defined
by $\M_{u,u}^{(t)}:=1/2$, $\M_{v,v}^{(t)}:=1/2$ and $\M_{u,v}^{(t)}=\M_{v,u}^{(t)}:=1/2$ if $\{u,v\} \in \M^{(t)} \subseteq E$,
and $\mathbf{M}^{(t)}_{u,u}=1$, $\mathbf{M}^{(t)}_{u,v}=0~(u\neq v)$ if $u$ is not matched.

\textbf{Balancing Circuit.} In the {\em balancing circuit model}, a specific sequence of matchings is applied periodically.
More precisely, let $\mathbf{M}^{(1)},\dots,\mathbf{M}^{(d)}$ be a sequence of $d$ matching matrices, also called {\em period} \footnote{Note that $d$ may be different from the maximal degree (or degree) of the underlying graph.}.
Then in step $t \geq 1$, we apply the matching matrix
$\mathbf{M}^{(t)} := \mathbf{M}^{(((t-1) \mod d)+1)}$.
We define the {\em round matrix} by $\M:= \prod_{s=1}^{d} \M^{(s)}$.  If $\M$ is symmetric, we define $\lambda(\M)$ to be its second largest eigenvalue (in absolute value). Following~\cite{RSW98}, if $\M$ is not symmetric (which is usually the case), we define $\lambda(\M)$ as the second largest eigenvalue of the symmetric matrix $\M \cdot \M^{T}$, where $\M^{T}$ is the transpose of $\M$. We always assume that $\lambda(\M) < 1$, which is guaranteed to hold if the matrix $\M$ is irreducible. Notice that since $\M$ is doubly stochastic, all powers of $\M$ are doubly stochastic as well. A natural choice for the $d$ matching matrices is given by an edge coloring of $G$. There are various efficient distributed edge coloring algorithms, e.g.~\cite{PS92,PS97}.


\textbf{Balancing Circuit on Specific Toplogies.}
For {\em hypercubes}, the canonical choice is dimension exchange consisting of $d=\log_2 n$ matching matching matrices $\M^{(i)}$ by $\M_{u,v}^{(i)}=1/2$ if and only if the bit representation of $u$ and $v$ differ only in bit $i$. Then the round matrix $\M$ is defined by $\prod_{i=1}^{\log_2 n} \M^{(i)}$. For {\em cycles}, we will consider the natural ``Odd-Even''-scheme meaning that for $\M^{(1)}$, the matching consists of all edges $\{j,(j+1) \pmod n \}$ for any odd $j$, while for $\M^{(2)}$, the matching consists of all edges $\{j,(j+1) \pmod n \}$ for any even $j$. More generally, for $r$-dimensional tori with vertex set $[0,n^{1/r}-1]^r$, we will have $2\cdot r$ matchings in total, meaning that for every dimension $1 \leq i \leq r$ we have two matchings along dimension $i$, similar to the definition of matchings for the cycle. 


\textbf{The Continuous Case.}
In the continuous case, load is arbitrarily divisible. Let $\xi^{(0)} \in \mathbb{R}^n$ be the initial load represented as a row vector, and in every round two matched nodes average their load perfectly. We consider the load vector $\xi^{(t)}$ after $t$ rounds in the balancing circuit model (that means, after the executions of $t \cdot d$ matchings in total). This process corresponds to a linear system and $\xi^{(t)},$ $t\in \mathbb{N}$, can be expressed as
$\xi^{(t)} = \xi^{(t-1)} \,\M$,
which results in $\xi^{(t)} = \xi^{(0)} \, \M^{t}$. 

\textbf{The Discrete Case.}
Let us now turn to the discrete case with indivisible, unit-size tokens. Let $x^{(0)} \in \mathbb{Z}^n$ be the initial load vector with average load $\overline{x}:=\sum_{w \in V} x_w^{(0)}/n$,
 and $x^{(t)}$ be the load vector at the end of round $t$. In case the sum of tokens of the two paired nodes is odd, we employ the so-called {\em random orientation} (or {\em randomized rounding})~\cite{RSW98,SS12}. More precisely, if there are two nodes $u$ and $v$ with load $a$ and $b$ being paired by matching $\M^{(t)}$, then node $u$ 
   gets either $\big\lceil \frac{a+b}{2} \big\rceil$ or $\big\lfloor \frac{a+b}{2} \big\rfloor $
tokens, with probability $1/2$ each. The remaining tokens are assigned to node $v$.

\textbf{The Average-Case Setting.}
We consider a setting where each entry of the initial load vector $x^{(0)}$ is chosen from an exponentially concentrated probability distribution $D$ with expectation $\mu$ and variance $\sigma^2$ (see Definition~\ref{def:expcon}).  
%
It is not difficult to verify that many natural distributions satisfy the condition of exponentially concentrated (see the appendix for more details).
\begin{restatable}{lem}{distributions}\label{lem:distributions}
The uniform distribution, binomial distribution, geometric distribution and Poisson distribution are all exponentially concentrated.
\end{restatable}

\begin{proof}
	Note that the uniform distribution $\Uni[0,k]$ is trivially exponentially concentrated, since $\sigma=\Theta(k)$. However, also distributions with unbounded range may be exponentially concentrated, with one example being the geometric distribution $\Geo(p)$. To verify this, first note that we have $\mu=1/p$ and $\sigma=\sqrt{(1-p)/p^2}$ (and so $\mu=\Theta(\sigma)$) and thus $\Pro{ \mu - X \geq \delta \cdot \sigma} \leq \exp\left( - \kappa \delta \right)$ holds trivially for a sufficiently small constant $\kappa > 0$. Secondly, for the upper tail, by Markov's inequality, $\Pro{ X \geq 2 \cdot \Ex{X} } \leq 1/2$, and by the memoryless property of the geometric distribution, for any $j \geq 1$, $\Pro{X \geq j \cdot 2 \cdot \Ex{X} } \leq 2^{-j}$. 
	
	For the binomial distribution $\Bin[m,p]$ with expectation $\mu=m \cdot p$ and standard deviation $\sigma=\sqrt{m \cdot p \cdot (1-p)}$, we will assume w.l.o.g. that $p \leq 1/2$, so that $\sigma = \Theta( \sqrt{mp})$. Then by \cite[Theorem~2.3]{McD98}, we have for $X \sim \Bin[m,p]$, $\Pro{X - \mu \geq  \epsilon \cdot \mu } \leq \exp\left(-  \frac{\epsilon^2 \mu}{2+2\epsilon/3} \right)$. Choosing $\epsilon = \delta \cdot \sigma/\mu$ yields
	$\Pro{X - \mu \geq \delta \cdot \sigma } \leq \exp\left(-  \frac{\delta^2 \cdot \sigma^2 / \mu}{2+2 \sigma/\mu} \right)$, as needed. For the lower tails, we use $\Pro{ \mu - X \geq \epsilon \cdot \mu} \leq e^{-1/2 \epsilon^2 \mu}$ and obtain a similar result as before (see again \cite[Theorem~2.3]{McD98}).
	
	For the Poisson Distribution $\Poi[\mu]$, we can verify in an analogous way that it is exponentially distributed by using the following two Chernoff bounds for Poisson random variables (\ref{lem:chernoffpoisson}).
\end{proof}

The definition of exponentially concentrated implies the following concentration result:

\begin{lem}\label{lem:distribution}
Let $D$ be an exponentially concentrated distribution and let $X \sim D$. Then,
\[
\Pro{ X \in [\mu - 8/\kappa \cdot \sigma \log n, \mu + 8/\kappa \cdot \sigma \log n] } \geq 1 -n^{-2}.
\]
In particular, the initial discrepancy satisfies $\disc(x^{(0)})=O( \sigma \cdot \log n)$ with probability at least $1-n^{-1}$.
\end{lem}



The advantage of Lemma~\ref{lem:distribution} is that we can use a simple conditioning trick to work with distributions that have a finite range and are therefore easier to analyze with concentration tools like Hoeffding's inequality (Theorem~\ref{thm:hoeffding}). That is in the analysis we simply work with a bounded-range distribution $\widetilde{D}$, which is the distribution $D$ under the condition that only values in the interval $[\mu - 8/c \cdot \sigma \log(n), \mu + 8/c \cdot \sigma \log(n)]$ occur.



\section{Proof of the General Bound (Theorem~\ref{thm:main1})}




\mainone*

\subsection{Proof of Theorem~\ref{thm:main1} (Upper Bound)}

We will use the following result from \cite{SS12} that bounds the deviation between the continuous and discrete load, assuming that we have $\xi^{(0)}=x^{(0)}$.

\begin{thm}[{\cite[Theorem 3.6($i$)]{SS12}}] \label{thm:errorbound}
	Consider the balancing circuit model with an arbitrary round matrix $\M$. Then for any round $t \geq 1$ it holds that
	\begin{equation*}
	\Pro{\max_{w \in V} \left| x_w^{(t)} - \xi_w^{(t)} \right| \leq \sqrt{12 \cdot \log n} } \geq 1 - n^{-2}.
	\end{equation*}
\end{thm}


\begin{proof}[Proof of Theorem~\ref{thm:main1} (Upper Bound)]

%
%

 Recall that the initial vector  $\xi^{(0)}=x^{(0)}$ consists of $n$ i.i.d.~random variables. As explained at the end of Section~\ref{sec:def}, we condition on the event \[ \mathcal{E} := \bigcap_{w \in V} \left\{ \left| \xi_w^{(0)} - \mu \right|  \leq 8/\kappa \cdot \sigma \cdot \log n \right\}. \]
By Lemma~\ref{lem:distribution}, $\Pro{ \mathcal{E} } \geq 1- n^{-2}$. In the remainder of the proof, all random variables are conditional on $\mathcal{E}$, but for simplicity we will not explicitly express this conditioning. 

Since $\xi_u^{(t)} = \sum_{w \in V} \xi_w^{(0)}\M^t_{w,k}$, the load $\xi_u^{(t)}$ is just a weighted sum of i.i.d.~random variables and we obtain
	\begin{align*}
		\xi_u^{(t)} - \xi_{v}^{(t)} &= \sum_{w \in V}\xi_w^{(0)} \cdot \left(\M^t_{w,u} - \M^t_{w,v}\right), \label{eq:deviation}
	\end{align*}
	which is in fact still a sum of $n$ i.i.d.~random variables. The expectation is
	\begin{align*}
	\Ex{\xi_u^{(t)} - \xi_{v}^{(t)}} &= \Ex {\sum_{w\in V} \xi_w^{(0)} \cdot \left(\M^t_{w,u} - \M^t_{w,v} \right) } = \sum_{w \in V} \left(\M^t_{w,u} - \M^t_{w,v} \right) \Ex{\xi_w^{(0)} } = 0,
	\end{align*}
	where the last equality holds since $\M$ is doubly stochastic.
	
	Now applying Hoeffding's inequality (Theorem~\ref{thm:hoeffding}) and recalling that conditional on $\mathcal{E}$, the range of each $\xi_w^{(t)}$ is $16/ \kappa \cdot \sigma \cdot \log n$, we obtain that
	\begin{equation*}
	\begin{aligned}
	\Pro{\left| \xi_u^{(t)} - \xi_{v}^{(t)} \right| \geq \delta} 
	&\leq 2 \cdot \exp \left(\frac{-2\delta^2}{256/\kappa^2 \cdot \sigma^2 \cdot \log^2 n \cdot 
		\left\| \M^t_{.,u} - \M^t_{.,v} \right\|_2^2 } \right).\\
	\end{aligned}
	\end{equation*}
	Applying Theorem~\ref{thm:errorbound} yields
	\begin{align*}
	\Pro{\left| x_u^{(t)} - x_{v}^{(t)} \right| \geq \delta + \sqrt{48 \cdot \log n}} \leq  2 \cdot \exp \left(\frac{-2\delta^2}{256/\kappa^2 \cdot \sigma^2 \cdot \log^2 n \cdot 
		\left\| \M^t_{.,u} - \M^t_{.,v} \right\|_2^2 } \right) + n^{-2}.
	\end{align*}
	The statement of the theorem follows by scaling $\delta$ and recalling that $\Pro{ \mathcal{E}} \geq 1 -n^{-2}$.
\end{proof}
%
%

\subsection{Proof of Theorem~\ref{thm:main1} (Lower Bound)}


The proof of the lower bound will use the following quantitative version of a central limit type theorem for independent but non-identical random variables.

\begin{thm}
	[Berry-Esseen Theorem \cite{berry1941accuracy,esseen1942liapounoff} for non-identical r.v.]\label{thm:be} Let $X_1,X_2,...,X_n$ be independently distributed with $\Ex{X_i} = 0$,
	$\Ex{X_i^2} = \Var{X_i} = \sigma_i^2$, and $\Ex{|X_i|^3} = \rho_i < \infty$. If $F_n(x)$ 
	is the
	distribution of $\frac{X_1 + ... + X_n}{\sqrt{\sigma_1^2 + \sigma_2^2 + ... + \sigma_n^2}}$ and $\Phi(x)$
	is the standard normal distribution, then 
	\begin{equation*}
	\begin{aligned}
	|F_n(x) - \Phi(x)| & \leq C_0 \psi_0,\\
	\end{aligned}
	\end{equation*}
	where
	$
	\psi_0 = \left( \sum_{i=1}^{n} \sigma_i^2 \right)^{-3/2} \cdot \sum_{i=1}^{n}\rho_i
	$
	and $C_0>0$ is a constant.
	\label{thm:berryesseenfornonidentical}
\end{thm}

With this concentration tool at hand, we are able to prove the lower bound in Theorem~\ref{thm:main1}. Unfortunately, it appears quite difficult to apply Theorem~\ref{thm:be} directly to equation~(\ref{eq:deviation}), since we need a good bound on the error term $\psi_0$. To this end, we will first partition the vertex set $V$ into buckets with equal contribution to $\xi_u^{(t)} - \xi_v^{(t)}$. Then we will apply Theorem~\label{thm:berryesseenfornonidentical} to the bucket with the largest variance, for which we can show that $\psi_0=o(1)$, thanks to the precondition that $t=\omega(1)$ and the bound in Theorem~\ref{thm:markov}. 

\begin{proof}[Proof of Theorem \ref{thm:main1} (Lower Bound)]
As in the derivation of the upper bound, we first consider $\xi_{u}^{(t)} - \xi_{v}^{(t)}$:
\begin{equation*}
	dev := \xi_u^{(t)} - \xi_{v}^{(t)} = \sum_{w \in V} \xi_w^{(0)} \cdot \left( \M^t_{w,u} - \M^t_{w,v} \right).
\end{equation*}
Again we are dealing with a weighted sum of i.i.d.~random variables with expectation $\mu$ and variance $\sigma^2$. As mentioned earlier, we have
$\Ex{dev} = \sum_{w \in V} \Ex{ \xi_w^{0}} \cdot 
\left( \M^t_{w,u} - \M^t_{w,v} \right) = 0$ since $\M$ is a doubly stochastic matrix. Of course, we could apply Theorem~\ref{thm:be} directly to $dev$, but it appears difficult to control the error term $\psi_0$. Therefore we will first partition the above sum into buckets where the weights of the random variables are roughly the same.

More precisely, we will partition $V$ into $2\log_2 \sigma - 1$ buckets, where for each $i$ we have $V_i := \left\{ w \in V \colon | \M^t_{w,u} - \M^t_{w,v}| \in (2^{-i},2^{-i+1}] \right\}$ for $1 \leq i \leq 2 \log_2 \sigma - 1$, and $V_{2\log_2 \sigma - 1} := \left\{ w\in V \colon \left| \M^t_{w,u} - \M^t_{w,v} \right| \leq \frac{1}{\sigma^2} \right\}$.

Further, let us consider the variance of $dev$:
\[
\sigma^2 = \sum_{w \in V} \left( \M^t_{w,k} - \M^t_{w,k'} \right)^2.
\]
Then by the pigeonhole principle there exists an index $1 \leq i \leq 2 \log_2 \sigma -1$ such that
\[
\sum_{v \in V_i} \left( \M^t_{w,u} - \M^t_{w,v} \right)^2 \geq \frac{1}{2 \log_2 \sigma} \cdot \sigma^2. 
\]
Firstly, if that index $i$ is equal to $2 \log_2 \sigma $, then 
\[
\left\| \M_{.,u}^t  - \M_{.,v}^t \right\|_2^2 = O(\sigma^{-1}), 
\]
and the lower bounds holds trivially. Therefore, we will assume in the remainder of the proof that $i < 2 \log_2 \sigma - 1$. We now decompose $dev$ into $dev=S+S^c$, where
\[
S := \sum_{w \in V_i} \xi_w^{(0)} \cdot \left( \M^t_{w,u} - \M^t_{w,v} \right)
\]
and
\[
S^c := \sum_{w \not\in V_i} \xi_w^{(0)} \cdot \left( \M^t_{w,u} - \M^t_{w,v} \right).
\]

Let us first analyze $S$. We will now apply Theorem~\ref{thm:be} to $S$. In preparation for this, let us first upper bound $\psi_0$. Using the definition of exponentially concentrated, it follows that for any constant $k$, the first $k$ moments are all bounded from above by $O(\sigma^k)$.
Hence,
\begin{align*}
	\psi_0 &= \frac{ \sum_{w \in V_i} \Ex{  \left|\xi_w^{(0)} \cdot \left( \M^t_{w,u} - \M^t_{w,v} \right)\right|  ^3}   }{ \left( \sum_{w \in V_i} \Ex{  \xi_w^{(0)} \cdot \left( \M^t_{w,u} - \M^t_{w,v} \right)  ^2}  \right)^{3/2} } \leq \frac{ O(\sigma^3) \cdot \sum_{w \in V_i} \left| \M^t_{w,u} - \M^t_{w,v} \right|  ^3}  { \sigma^3 \cdot \left( \sum_{w \in V_i}  \left( \M^t_{w,u} - \M^t_{w,v} \right)  ^2  \right)^{3/2}}.
\end{align*}
Recalling that for any $w \in V_i$, $\left| \M^t_{w,u} - \M^t_{w,v}   \right| \in (2^{-i},2^{-i+1}] $, we can simplify the above expression as follows:
\begin{align*}
	\psi_0 &= O \left(	\frac{  |V_i| \cdot 2^{-3i}  }{ |V_i|^{3/2} \cdot 2^{-3i} }		\right) = O( |V_i|^{-1/2}).
\end{align*}
However, since we have $t=\omega(1)$, by Theorem~\ref{thm:markov}, $| \M^t_{x,y} - \frac{1}{n} | = O(t^{-1/2})$ and therefore it must be that $|V_i|=\omega(1)$, and we conclude that $\psi_0 = o(1)$. 

Before applying Theorem~\ref{thm:be}, we scale the original distribution to $\xi_w^{'(0)} = \xi_w^{(0)} - \mu$. Since $\Var(aX) = a^2 \Var(X)$, we have
\begin{equation*}
	\begin{aligned}
		F_n(x) & = \Pro{ \frac{\sum_{w \in V_i} \xi_w^{'(0)} \cdot \left( \M^t_{w,u} - \M^t_{w,v} \right) }{\sigma \sqrt{\sum_{w \in V_i} \left( \M^t_{w,u} - \M^t_{w,v} \right)^2}} \leq x } \\
		& = \Pro{ \frac{\sum_{w \in V_i} \xi_w^{(0)} \cdot \left( \M^t_{w,u} - \M^t_{w,v} \right) - \sum_{w \in V_i}  \mu \cdot \left( \M^t_{w,u} - \M^t_{w,v} \right) }{\sigma \sqrt{\sum_{w \in V_i}\left( \M^t_{w,u} - \M^t_{w,v} \right)^2}} 
			\leq x } \\
		& = \Pro{ \frac{S - \Ex{S}}{\sigma \sqrt{\sum_{w \in V_i}\left(\M^t_{w,u} - \M^t_{w,v} \right)^2}} 
			\leq x } \\
		& = \Pro{ S - \Ex{S} 
			\leq x\sigma \sqrt{\sum_{w \in V_i}\left(\M^t_{w,u} - \M^t_{w,v} \right)^2} }.
	\end{aligned}  
\end{equation*}

As derived earlier $\psi_0 = o(1)$, and therefore

\begin{align*}
	\Pro{  S - \Ex{S}  \geq x \sigma \sqrt{\sum_{w \in V_i}\left( \M^t_{w,u} - \M^t_{w,v} \right)^2} } & \geq  \Phi(-x) - C_0\psi_0  \\
	& \geq \frac{1}{\sqrt{\pi} (x + \sqrt{x^2 + 2})e^{x^2}} - o(1),
\end{align*}
where last inequality uses \cite[Formula~ 7.1.13]{abramowitz1966handbook}:
\begin{equation*}
	\frac{1}{x + \sqrt{x^2 + 2}} <
	e^{x^2}\int_x^\infty e^{-t^2}dt \leqslant
	\frac{1}{x + \sqrt{x^2 + 4/\pi}} \ (x>0).
\end{equation*}

Therefore, by substitution, we get
\begin{equation*}\label{fm:normesti}
	\frac{1}{\sqrt{\pi}(x+\sqrt{x^2+2})e^{x^2}} <
	\Phi^c(x) \leqslant
	\frac{1}{\sqrt{\pi}(x+\sqrt{x^2+4/\pi})e^{x^2}}.
\end{equation*}

Hence with $x=1$,
\begin{align*}
	\Pro{ S - \Ex{S} \geq \sigma \sqrt{\sum_{w \in V_i}\left( \M^t_{w,u} - \M^t_{w,v} \right)^2} } &\geq \frac{1}{16}.
\end{align*}

Similarly, we can derive that
\begin{align*}
	\Pro{ \Ex{S} - S \geq \sigma \sqrt{\sum_{w \in V_i}\left( \M^t_{w,u} - \M^t_{w,v} \right)^2} } &\geq \frac{1}{16}.
\end{align*}

Hence, independent of what the value $S^c$ is, there is still a probability of at least $1/16$ so that $|S+S^c| \geq \sigma/2 \cdot \sqrt{1/(2 \log_2 \sigma)} \cdot \sqrt{\sum_{w \in V}\left( \M^t_{w,u} - \M^t_{w,v} \right)^2}$.
\end{proof}

\section{Proof of the Universal Bounds (Theorem~\ref{thm:maintwo}, Theorem~\ref{thm:main3})}\label{sec:universalbounds}

In the previous section we proved that the deviation between the loads of two nodes $u$ and $v$ is essentially captured by $\left\| \M^t_{.,u} - \M^t_{.,v} \right\|_2$. However, in some cases it might be hard to compute or estimate this quantity for arbitrary vertices $u$ and $v$. Therefore we will first prove the following universal upper bound on the discrepancy that works for arbitrary graphs and pair of nodes, as stated on page~\pageref{thm:maintwo}. 


\maintwo* 

\subsection{Proof of Theorem~\ref{thm:maintwo}}

The proof of Theorem~\ref{thm:maintwo} is fairly involved and we first sketch the high level ideas.
We first show that $\left\| \M^t_{.,u} - \M^t_{.,v} \right\|_2^2$ can be upper bounded in terms of the $\ell_2$-distance to the stationary distribution.


\begin{restatable}{lem}{twonodestoavg}
\label{lem:twonodestoavg}
	Consider the balancing circuit model with an arbitrary round matrix $\M$. Then for all $u,v \in V$, we have
	$
	 \| \M^t_{.,u} - \M^t_{.,v}  \|_2^2 \leq 4 \cdot \max_{k \in V} \| \M^t_{.,k} - \mathbf{\frac{1}{n}}  \|_2^2.
$
	Further, for any $u \in V$ we have
	$
	\max_{v \in V} \| \M^t_{.,u} - \M^t_{.,v} \|_2^2 \geq \| \M^t_{.,u} - \mathbf{\frac{1}{n}} \|_2^2
	$.
\end{restatable}

\begin{proof}
	\begin{equation*}
		\begin{aligned}
			\sum_{w \in V}\left( \M^t_{w,u} - \M^t_{w,v} \right)^2 & \leq 2 \cdot \left( \sum_{w \in V} \left( \M^t_{w,u} - \frac{1}{n} \right)^2 + \left( \M^t_{w,v} - \frac{1}{n} \right)^2 \right) \\
			& \leq 4 \cdot \max_{k \in V} \sum_{w \in V}\left( \M^t_{w,k} - \frac{1}{n} \right)^2,
		\end{aligned}
	\end{equation*}
	and the first statement follows. 
	We now prove the second statement:
	\begin{equation*}
		\forall u \colon \max_{v \in V} \sqrt{\sum_{w \in V} \left(\M_{w,u}^t - \M_{w,v}^t \right)^2} \geq \sqrt{\sum_{w \in V} \left(\M_{w,u}^t - \frac{1}{n} \right)^2}.
	\end{equation*}
	
	We first look at the difference between these two terms squared. That is, for any vertex $v \in V$ we have
	\begin{align}
		\lefteqn{ \sum_{w \in V} \left(\M_{w,u}^t - \M_{w,v}^t \right)^2   - \sum_{w \in V} \left(\M_{w,u}^t - \frac{1}{n} \right)^2} \notag \\
		& = -2 \sum_{w \in V} \left(\M_{w,u}^t - \frac{1}{n}\right)\left(\M_{w,v}^t - \frac{1}{n}\right) + \sum_{w \in V} \left(\M_{w,v}^t - \frac{1}{n}\right)^2 \notag \\
		& = -2 \sum_{w \in V}  \M_{w,u}^t \cdot \M_{w,v}^t + \frac{4}{n} - \frac{2}{n} + \sum_{w \in V} \left(\M_{w,v}^t\right)^2 + \frac{1}{n} \notag \\
		& = -2 \sum_{w \in V}  \M_{w,u}^t \cdot \M_{w,v}^t + \sum_{w \in V} \left(\M_{w,v}^t\right)^2 + \frac{1}{n} \label{eq:labelone}
	\end{align}
	
	Now let $Z$ be a uniform random variable over the set $V \backslash \{u\}$. Then it follows that
	
	\begin{equation*}
		\begin{aligned}
			\Esub{Z \sim V \setminus \{u\}}{\sum_{w \in V} \M_{w,u}^t \cdot \M_{w,Z}^t } &
			= \sum_{z \in V, z \neq u} \frac{1}{n-1} \cdot \sum_{w \in V} \M_{w,u}^t \cdot \M_{w,z}^t \\
			& = \frac{1}{n-1} \sum_{w \in V} \M_{w,u}^t \cdot \sum_{z \in V, z \neq u}\M_{w,z}^t \\
			& = \frac{1}{n-1} \sum_{w \in V} \M_{w,u}^t \cdot \left( 1 - \M_{w,u}^t \right) \\
			& = \frac{1}{n-1} \left( 1 - \sum_{w \in V} \left(\M_{w,u}^t\right)^2 \right).
		\end{aligned}
	\end{equation*}
	Further, by linearity of expectations
	\begin{align*}
		\lefteqn{ \Esub{Z \sim V \setminus \{u\}}{-2 \sum_{w \in V} \M_{w,u}^t \cdot \M_{w,Z}^t + \sum_{w \in V} \left( \M_{w,Z}^t \right)^2 + \frac{1}{n} }} \\ &= \frac{-2}{n-1} \left( 1 - \sum_{w \in V} \left(\M_{w,u}^t\right)^2 \right) + \sum_{z \in V, y \neq u} \sum_{w \in V} \frac{1}{n-1} \left( \M_{w,z}^t \right)^2 + \frac{1}{n}.
	\end{align*}

	By definition of expectation, this implies that there exists a vertex $v \in V, v \neq u$ such that
	\begin{align}
		\lefteqn{-2 \sum_{w \in V} \M_{w,u}^t \cdot \M_{w,v}^t + \sum_{w \in V} \left( \M_{w,v}^t \right)^2 + \frac{1}{n} } \notag \\ &\geq 	\frac{-2}{n-1} \left( 1 - \sum_{w \in V} \left(\M_{w,u}^t\right)^2 \right)	+ \sum_{z \in V, z \neq u} \sum_{w \in V} \frac{1}{n-1} \left( \M_{w,z}^t \right)^2 + \frac{1}{n}. \label{eq:labeltwo}
	\end{align}
	
	
	Combining (\ref{eq:labelone}) and (\ref{eq:labeltwo}), 
	\begin{align*}
		\lefteqn{ \sum_{w \in V} \left(\M_{w,u}^t - \M_{w,v}^t \right)^2  - \sum_{w \in V} \left(\M_{w,v}^t - \frac{1}{n} \right)^2 } \\ & \geq \frac{-2}{n - 1} \cdot \left( 1- \sum_{w \in V} \left(\M_{w,u}^t \right)^2 \right) + \sum_{w \in V} \sum_{z \in V, z \neq u} \frac{1}{n-1} \cdot \left( \M_{w,z}^t \right)^2+ \frac{1}{n} \\
		& = -\frac{1}{n-1} - \frac{1}{n \cdot (n-1)} + \frac{2}{n-1}\sum_{w \in V} \left( \M_{w,u}^t \right)^2  + \frac{1}{n - 1} \cdot \sum_{w \in V} \sum_{z \in V, z \neq u}      \left(\M_{w,z}^t \right)^2\\
		&= -\frac{1}{n-1} - \frac{1}{n \cdot (n-1)} + \frac{1}{n-1}\sum_{w \in V} \left(\M_{w,u}^t\right)^2 + \frac{1}{n-1} \cdot \sum_{w \in V} \sum_{z \in V}  \left(\M_{w,z}^t \right)^2\\
		& \geq -\frac{1}{n-1} - \frac{1}{n \cdot (n-1)} + \frac{1}{n \cdot (n-1)} + \frac{1}{n-1} \\
		& \geq 0,
	\end{align*}
	where the second last inequality holds since $\M$ is doubly stochastic.
\end{proof}


The next step and main ingredient of the proof of Theorem~\ref{thm:maintwo} is to establish that $\| \M^t_{.,k} - \mathbf{\frac{1}{n}} \|_{\infty} = O(1/\sqrt{t})$. This result will be a direct application of a general bound on the $t$-step probabilities of an arbitrary, possibly non-reversible Markov chain, as given in Theorem~\ref{thm:markov} from page \pageref{thm:markov}:

\markov*

In this subsection we prove Theorem~\ref{thm:maintwo}, assuming the correctness of Theorem~\ref{thm:markov} whose proof is deferred to Section~\ref{sec:markovproof}.

\begin{proof}[Proof of Theorem~\ref{thm:maintwo}]

	By Theorem \ref{thm:main1} and Lemma \ref{lem:twonodestoavg}, we obtain 
	\begin{align*}
	\Pro{\left| x_u^{(t)} - x_{v}^{(t)} \right| \geq  \delta  \cdot 16\sqrt{2} \kappa \cdot \sigma \cdot \log n \cdot
		\max_{k \in V} \left\| \M^t_{.,k} - \mathbf{\frac{1}{n}}  \right\|_2  + \sqrt{48 \log n}} \leq 2 \cdot e^{- \delta^2} + 2n^{-2}.
	\end{align*}
	Hence we can find a $\delta=\sqrt{3\log n}$ so that the latter probability gets smaller than $3 n^{-2}$. Further, by applying Theorem \ref{thm:markov} with $\alpha = \beta = 2^{-d}$ to $\P=\M$ we conclude that
	$
	\| \M^t_{.,k} - \mathbf{\frac{1}{n}} \|_{\infty} = O(t^{-1/2}), 
	$
	since $d=O(1)$. Using the fact $\|.\|_2^2 \leq \|.\|_{\infty} \cdot \|.\|_1$,
	$
	\| \M^t_{.,k} - \mathbf{\frac{1}{n}} \|_{2}^2 = O(t^{-1/2}),
	$
	and by the union bound, $\disc(x^{(t)}) = O(t^{-1/4} \cdot \sigma \cdot (\log n)^{3/2} + \sqrt{\log n})$ with probability at least $1-3 n^{-1}$.
%
%
%
\end{proof}

\subsection{Proof of Theorem~\ref{thm:markov}}\label{sec:markovproof}

This section is devoted to the proof of Theorem~\ref{thm:markov}. Our proof is based on the evolving-set process, which is a Markov chain based on any given irreducible, not necessarily reversible Markov chain on $\Omega$. For the definition of the evolving set process, we closely follow the exposition in \cite[Chapter 17]{levin2009markov}.

Let $\P$ denote the transition matrix of an irreducible Markov chain and $\pi$ its stationary distribution. $\P^t$ is the $t$-step transition probability matrix. The \emph{edge measure} $Q$ is defined by $Q_{x,y} := \pi_x \P_{x,y}$ and $Q(A, B) = \sum_{x \in A, y \in B} Q_{x,y}$. 





\begin{defi}\label{def:evolset}
	Given a transition matrix $\P$, the \textbf{evolving-set process} is a Markov chain on subsets of $\Omega$ defined as follows. Suppose the current state is $S \subset \Omega$. Let $U$ be a random variable which is uniform on $[0,1]$. The next state of the chain is the set
	\begin{equation*}
	\tilde{S} = \left\{ y \in \Omega : \frac{Q(S,y)}{\pi_y} \geq U \right\}.
	\end{equation*}
\end{defi}

This chain is \emph{not irreducible} because 
$\varnothing$ and $\Omega$ are absorbing states.
It follows that
\begin{equation*}
\Pro{y \in S_{t+1} \,|\, S_t} = \frac{Q(S_t, y)}{\pi_y}
\end{equation*}
since the probability that $y \in S_{t+1}$ is equal to the probability of the event that the chosen value of $U$ is less than $\frac{Q(S_t, y)}{\pi_y}$. 


\begin{proposition}[{\cite[Proposition 17.19]{levin2009markov}}]\label{prop:bound}
	Let $(M_t)$ be a non-negative martingale with respect to $(Y_t)$, and define
	$
	T_h := \min\{t \geq 0: M_t = 0 \ or \ M_t \geq h\}
	$
	Assume that for any $h \geq 0$
	\begin{enumerate}[(i)]\itemsep0pt
		\item For $t < T_h$, $\Var(M_{t+1} \,|\, Y_0, \ldots , Y_t) \geq \sigma^2$, and
		\item $M_{T_h} \leq Dh$.
	\end{enumerate}
	Let $T := T_1$. If $M_0$ is a constant, then
	$
	\Pro{T > t} \leq \frac{2M_0}{\sigma}\sqrt{\frac{D}{t}}.
	$
\end{proposition}

We now generalize \cite[Lemma~17.14]{levin2009markov} to cover arbitrarily small loop probabilities. 


\begin{lem}\label{lem:generalisedexpi}
	Let $(U_t)$ be a sequence of independent random variables, each uniform on $[0,1]$, such that $S_{t+1}$ is generated from $S_t$ using $U_{t+1}$. Then with $\beta := \min\limits_{u} \P_{u,u} > 0$,
	\begin{equation*}
	\begin{aligned}
	\Ex{\pi(S_{t+1}) \,|\, U_{t+1} \leq \beta, S_t = S} & \geq \pi(S) + Q(S, S^c)
	, \\
	\Ex{\pi(S_{t+1}) \,|\, U_{t+1} > \beta, S_t = S} & \leq \pi(S) - \frac{\beta Q(S,S^c)}{1 - \beta}.
	\end{aligned}
	\end{equation*}
\end{lem}

We list a few auxiliary results from \cite{levin2009markov} about the evolving set process that will be used to prove the result.

\begin{lem}[{\cite[Lemma 17.12]{levin2009markov}}]\label{lem:ptoevolp}
	If $(S_t)_{t \geq 0}$ is the evolving-set process associated to the transition matrix $\P$, then for any time $t$ and $x,y \in \Omega$
	\begin{equation*}
		\P^t_{x,y} = \frac{\pi_y}{\pi_x}\Psub{\{x\}}{y \in S_t}.
	\end{equation*}
\end{lem}
Recall that $(S_t)$ is the evolving-set process based on the Markov chain whose transition matrix is $\P$. $\Psub{\{x\}}{y \in S_t}$ means the probability of the event $y \in S_t$ with the initial state of the evolving set being $\{ x \}$.

\begin{lem}[{\cite[Lemma 17.13]{levin2009markov}}]\label{lem:pistmartingale}
	The sequence $\{\pi(S_t)\}$ is a martingale.
\end{lem}

\begin{thm}[{\cite[Corollary 17.7]{levin2009markov}}]\label{lem:ost}
	Let $(M_t)$ be a martingale and $\tau$ a stopping time. If $\Pro{\tau < \infty}$ and $|M_{t \land \tau}| \leq K$ for all $t$ and some constant $K$ where $t \land \tau := \min\{t, \tau \}$, then $\Ex{M_{\tau}} = \Ex{M_0}$.
\end{thm}

\begin{proof}[Proof of Lemma \ref{lem:generalisedexpi}]
	Given $U_{t+1} \leq \beta$, the distribution of $U_{t+1}$ is uniform on $[0, \beta]$. 
	
	\textbf{Case 1: }For $y \notin S$, we know that for $y$ satisfying $\frac{Q(S,y)}{\pi_y} \in [0, \beta]$
	
	\begin{equation*}
		\Pro{\frac{Q(S,y)}{\pi_y} \geq U_{t+1} \,\Big|\, U_{t+1} \leq \beta, S_t = S} = \frac{Q(S,y)}{\beta\pi_y}
	\end{equation*}
	
	and for $y$ satisfying $\frac{Q(S,y)}{\pi_y} \in (\beta,1]$,
	
	\begin{equation*}
		\Pro{\frac{Q(S,y)}{\pi_y} \geq U_{t+1} \,\Big|\, U_{t+1} \leq \beta, S_t = S} = 1.
	\end{equation*}
	
	We know that
	
	\begin{equation*}
		\frac{Q(S,y)}{\pi_y} = \frac{\sum_{x \in S}\pi_x\P_{x,y}}{\pi_y} \leq \frac{\sum_{x \in \Omega}\pi_x\P_{x,y}}{\pi_y} = 1.
	\end{equation*}

	Since $y \in S_{t+1}$ if and only if $U_{t+1} \leq Q(S_t,y)/\pi_x$, we therefore can combine the above results by using an inequality and conclude that
	\begin{equation*}
		\Pro{y \in S_{t+1} \,|\, U_{t+1} \leq \beta, S_t = S } \geq \frac{Q(S,y)}{\pi_y} \text{ for $y \notin S$}
	\end{equation*}
	because $\beta \leq 1$ and $Q(S,y)/\pi_y \leq 1$.
	
	\textbf{Case 2:} For $y \in S$, we have $Q(S,y)/\pi_y \geq Q(y,y)/\pi_y \geq \beta$, it follows that when $U_{t+1} \leq \beta$
	\begin{equation*}
		\Pro{y \in S_{t+1} \,|\, U_{t+1} \leq \beta, S_t = S } = 1
		\text{  for $y \in S$}.
	\end{equation*}
	
	We have
	
	\begin{equation*}
		\begin{aligned}
			& \Ex{\pi(S_{t+1}) \,|\, U_{t+1} \leq \beta, S_t = S}  \\
			& = \Ex{\sum_{y \in \Omega} \mathbbm{1}_{\{y \in S_{t+1}\} } \pi_y \,\Big|\, U_{t+1} \leq \beta, S_t = S } \\
			& = \sum_{y \in S}\pi_y\Pro{y \in S_{t+1} \,|\, U_{t+1} \leq \beta, S_t = S}  + \sum_{y \notin S}\pi_y\Pro{y \in S_{t+1} \,|\, U_{t+1} \leq \beta, S_t = S}.
		\end{aligned}
	\end{equation*}

	Based on previous results, we can see that
	
	\begin{equation*}
		\Ex{\pi(S_{t+1}) \,|\, U_{t+1} \leq \beta, S_t = S} \geq \pi(S) + Q(S, S^c).
	\end{equation*}
	
	By Lemma \ref{lem:pistmartingale} and the formulas above,
	
	\begin{equation*}
		\begin{aligned}
			\pi(S) & = \Ex{\pi(S_{t+1}) \,|\, S_t = S} \\
			& = \beta \cdot \Ex{\pi(S_{t+1}) \,|\, U_{t+1} \leq \beta, S_t = S}  + (1 - \beta) \cdot\Ex{\pi(S_{t+1}) \,|\, U_{t+1} > \beta, S_t = S}.
		\end{aligned}
	\end{equation*}
	
	Rearranging shows that
	
	\begin{equation*}
		\Ex{\pi(S_{t+1}) \,|\, U_{t+1} > \beta, S_t = S} \leq \pi(S) - \frac{\beta Q(S,S^c)}{1 - \beta}.
	\end{equation*}
	
\end{proof}


The derivation of the next lemma closely follows the analysis in~\cite[Chapter 17]{levin2009markov}. For the sake of completeness, a proof can be found in the appendix.
\begin{restatable}{lem}{partone}\label{lem:partone}
For any two states $x,y$,
$
\left| \P^t_{x,y} - \pi_y \right| \leq \frac{\pi_y}{\pi_x} \cdot \Psub{\{x\}}{\tau > t}.
$
\end{restatable}

\begin{proof}
	First of all, let the hitting time
	\begin{equation*}
		\tau = \min\{t \geq 0 : S_t \in \{ \varnothing, \Omega\}  \}.
	\end{equation*}
	
	We have $S_{\tau} \in \{\varnothing, \Omega\}$ and $\pi(S_{\tau}) = \mathbbm{1}_{ \{S_{\tau} = \Omega \}}$. 
	We consider an evolving set process with $S_0 = \{x\}$. By Theorem \ref{thm:ost} and Lemma \ref{lem:pistmartingale},
	
	\begin{equation}\label{equ:pixequalspx}
	\begin{aligned}
	\pi_x & = \Esub{\{x\}}{\pi(S_0)} = \Esub{\{x\}}{\pi(S_{\tau})} = \Esub{\{x\}}{\mathbbm{1}_{ \{S_{\tau} = \Omega \}}} = \Psub{\{x\}}{S_{\tau} = \Omega} = \Psub{\{x\}}{x \in S_{\tau}} \\
	\end{aligned}
	\end{equation}
	
	For the last equality, it is true because we know that $S_{\tau}$ can only be $\varnothing$ or $\Omega$. Hence, the probability that $x$ is an element in $S_\tau$ is equal to the probability that $S_\tau$ is $\Omega$. Note that here the second $x$ in the last line can be any other element in $\Omega$. For example, we also know that 
	
	\begin{equation}\label{equ:Stau}
	\forall y \in \Omega, \Psub{\{x\}}{S_\tau = \Omega} = \Psub{\{x\}}{y \in S_\tau}.
	\end{equation}
	
	For our bound, we know that by Lemma \ref{lem:ptoevolp} and (\ref{equ:pixequalspx}),
	
	\begin{equation*}
		\begin{aligned}
			\left| \P^t(x,y) - \pi_y\right| & = \frac{\pi_y}{\pi_x}\left|\Psub{\{x\}}{y \in S_t} - \pi_x \right| = \frac{\pi_y}{\pi_x}\left|\Psub{\{x\}}{y \in S_t} - \Psub{\{x\}}{S_\tau = \Omega} \right|.
		\end{aligned}
	\end{equation*}
	
	By (\ref{equ:Stau}),
	\begin{equation*}
	\begin{aligned}
	\Psub{\{x\}}{y \in S_t} & = \Psub{\{x\}}{y\in S_t, \tau > t} + \Psub{\{x\}}{y\in S_t, \tau \leq t} \\
	& = \Psub{\{x\}}{y\in S_t, \tau > t} + \Psub{\{x\}}{S_\tau = \Omega, \tau \leq t}.
	\end{aligned}
	\end{equation*}
	
	By simple substitution we obtain
	
	\begin{equation*}
		\begin{aligned}
			\left| \P^t_{x,y} - \pi_y \right| & = \frac{\pi_y}{\pi_x}\left|\Psub{\{x\}}{y\in S_t, \tau > t} + \Psub{\{x\}}{S_\tau = \Omega, \tau \leq t} - \Psub{\{x\}}{S_\tau = \Omega} \right|\\
			& = \frac{\pi_y}{\pi_x}\left|\Psub{\{x\}}{y\in S_t, \tau > t} - \Psub{\{x\}}{S_\tau = \Omega, \tau > t} \right| \\
			& \leq  \frac{\pi_y}{\pi_x} \Psub{\{x\}}{\tau > t}.
		\end{aligned}
	\end{equation*}
	The last line is true because we remove all possible intersections. 
\end{proof}

	Now we want to use Proposition \ref{prop:bound} to bound $\Psub{\{x\}}{\tau > t}$. To apply it, we substitute the following parameters: $M_0$ is chosen to be $\pi(\{x\})$, $Y_t$ is $S_t$, and $T = T_1 := \min\{t \geq 0:  \pi(S_t) = 0 \ or\ \pi(S_t) \geq 1 \}$. Hence in our case, $\tau$ is the same as $T$ (or $T_1$) in the proposition. The following two lemmas elaborate on the two preconditions (i) and (ii) of Proposition~\ref{prop:bound}. 
	
	
		
	
	\begin{restatable}{lem}{parttwo}\label{lem:parttwo}
For any time $t$ and $S_0=\{x\}$,
$
	\Var_{S_t}(\pi(S_{t+1})) \geq \beta\pi_{\min}^2\alpha^2.
$
\end{restatable}
\begin{proof}
	Conditioning always reduces variance and $S_t \neq \varnothing $ or $ \Omega$, we have
	\begin{equation*}
		\Var_{S_t}\big(\pi(S_{t+1})\big) \geq \Var_{S_t}\Big(\Ex{\pi(S_{t+1}) \,|\, \mathbbm{1}_{\{U_{t+1} \leq \beta \}}}\Big).
	\end{equation*}	
	For $S_t = S$,
	\begin{equation*}
		\Esub{S_t}{\pi(S_{t+1}) \,|\, \mathbbm{1}_{\{U_{t+1} \leq \beta \}}} =
		\begin{cases}
			\Ex{\pi(S_{t+1}) \,|\, U_{t+1} \leq \beta, S_t = S} , & \text{w.p. } \beta, \\
			\Ex{\pi(S_{t+1}) \,|\, U_{t+1} > \beta, S_t = S}, & \text{w.p. } 1 - \beta
		\end{cases}
	\end{equation*}
	and by Lemma \ref{lem:generalisedexpi}, we know that
	
	\begin{equation*}
		\Esub{S_t}{\pi(S_{t+1}) \,|\, \mathbbm{1}_{\{U_{t+1} \leq \beta \}}}
		\begin{cases}
			\geq \pi(S) + Q(S, S^c) , & \text{w.p. } \beta, \\
			\leq \pi(S) - \frac{\beta Q(S,S^c)}{1 - \beta} , & \text{w.p. } 1 - \beta.
		\end{cases}
	\end{equation*}

	For simplicity, we let $\Esub{S_t}{\pi(S_{t+1}) \,|\, \mathbbm{1}_{\{U_t \leq \beta \}}}$ be $X$, $\Ex{\pi(S_{t+1}) \,|\, U_{t+1} \leq \beta, S_t = S}$ be $x_1$ and $\Ex{\pi(S_{t+1} \,|\, U_{t+1} > \beta, S_t = S)}$ be $x_2$. Then we have
	
	\begin{equation*}
		\begin{aligned}
			\Var_{S_t}& \Big(\Ex{\pi(S_{t+1}) \,|\, \mathbbm{1}_{\{U_t \leq \beta \}}}\Big) = \Ex{X^2} - \Ex{X}^2 \\
			& = \beta x_1^2 + (1 - \beta)x_2^2 - (\beta x_1 + (1-\beta)x_2)^2 = (\beta - \beta^2)(x_1 - x_2)^2
		\end{aligned}
	\end{equation*}
	
	In order to derive a lower bounds on this variance, based on Lemma \ref{lem:generalisedexpi} we let $x_1 = \pi(S) + Q(S,S^c)$ and $x_2 = \pi(S) - (\beta/1 - \beta)Q(S,S^c)$. With this we obtain
	
	\begin{equation*}
		\begin{aligned}
			\Var_{S_t} & \Big(\Ex{\pi(S_{t+1}) \,|\, \mathbbm{1}_{\{U_t \leq \beta \}}}\Big) \geq \frac{\beta}{1 - \beta}Q^2(S,S^c).
		\end{aligned}
	\end{equation*}

	Therefore, provided $S_t \notin \{\varnothing, \Omega \}$, we have
	
	\begin{equation*}
		\Var_{S_t}\Big( \pi(S_{t+1}) \Big) \geq \frac{\beta}{1 - \beta}Q^2(S, S^c) \geq \frac{\beta \pi_{\min}^2 \alpha^2}{(1 - \beta)}.
	\end{equation*}
	
	The last inequality follows from the fact that if $S \notin \{\varnothing, \Omega \}$ then there exist $u \in S, v \notin S$ with $\P_{u,v} > 0$, whence 
	\begin{equation*}
		Q(S, S^c) = \sum_{\substack{s \in S \\ w \in S^c}}\pi_s \P_{s,w} \geq \pi_u \P_{u,v} \geq \pi_{\min}\alpha.
	\end{equation*}
	
	Since $1 - \beta < 1$, we finally obtain
	\begin{equation*}
		\Var_{S_t}(\pi(S_{t+1})) \geq \beta\pi_{\min}^2\alpha^2. 
	\end{equation*}
\end{proof}

		Finally, we derive an upper bound on the amount by which $S_t$ can increase in one iteration.
	
		\begin{restatable}{lem}{partthree}\label{lem:partthree}
For any time $t$ and  $S_0=\{x\}$,
$
	\pi(S_{t+1}) \leq  \left(\frac{1-\beta}{\alpha} + 1\right)\frac{\pi_{\max}}{\pi_{\min}} \cdot \pi(S_{t}).
$
\end{restatable}

\begin{proof}
	Since
	\begin{equation*}
		S_{t+1} = \left\{ y \in \Omega: \frac{\sum_{x \in S_t} \pi_x \P_{x,y}}{\pi_y} \geq U \right\}.
	\end{equation*} 
	If $U$ decreases to $0$, then every $y \in S_{t+1}$ is at least connected to an $x \in S_t$. In other words, $\P_{x,y} > 0$ for $x \in S_t$ and $y \in S_{t+1}$. Hence $|S_{t+1}| \leq (\frac{1 - \beta}{\alpha} + 1)|S_t|$. 
	
	We also know that
	
	\begin{equation*}
		\begin{aligned}
			\pi(S_{t+1}) \leq |S_{t+1}| \cdot \pi_{\max} \leq \left(\frac{1 - \beta}{\alpha} + 1 \right) \cdot |S_t| \cdot \pi_{\max} \leq \left(\frac{1 - \beta}{\alpha} + 1 \right) \cdot \pi(S_t) \cdot \frac{\pi_{\max}}{\pi_{\min}}
		\end{aligned}
	\end{equation*}
\end{proof}

		The proof of Theorem~\ref{thm:markov}
		follows then by combining Proposition~\ref{prop:bound}, Lemma \ref{lem:generalisedexpi}, Lemma \ref{lem:partone}, Lemma \ref{lem:parttwo} and Lemma \ref{lem:partthree}. 
		
		\begin{proof}[Proof of Theorem~\ref{thm:markov}]
			With the help of the previous three lemmas, we can apply Proposition \ref{prop:bound} with $M_0 = \pi_x$, $\sigma \geq \beta^{1/2} \pi_x \alpha$ and $D = \left(\frac{1-\beta}{\alpha}\right) \frac{\pi_{\max}}{\pi_{\min}}$ to obtain
			
			\begin{equation*}
				\begin{aligned}
					|\P^t_{x,y} - \pi_y| & \leq \frac{\pi_y}{\pi_x}\Psub{\{x\}}{\tau > t}\\
					& \leq \frac{\pi_y}{\pi_x} \frac{2\pi_x}{\sigma}\sqrt{\frac{D}{t}} \\
					& = \frac{2\pi_y}{\beta^{1/2} \pi_{\min} \alpha}\sqrt{\frac{(\frac{1-\beta}{\alpha} + 1)\frac{\pi_{\max}}{\pi_{\min}}}{t}}\\
					& \leq \frac{\pi_{\max}^{3/2}}{\pi_{\min}^{3/2}}\cdot \frac{2}{\beta^{1/2}\alpha} \sqrt{\frac{1-\beta+\alpha}{\alpha t}} \\
				\end{aligned}
			\end{equation*}
			
		\end{proof}

\subsection{Proof of Theorem~\ref{thm:main3}}

We now prove the following discrepancy bound that depends on the $\lambda(\M)$, as defined in Section \ref{sec:def}.

\begin{proof}
By \cite[Lemma 2.4]{SS12}, for any pair of vertices $u,v \in V$,
$
\left| \M^t_{u,v} - \frac{1}{n} \right| \leq \lambda(\M)^{t/2}.
$
Hence by Lemma~\ref{lem:twonodestoavg} $\left\| \M^t_{.,u} - \M^t_{.,v} \right\|_2 = O(\lambda(\M)^{t/4})$, and the bound on the discrepancy follows from Theorem~\ref{thm:main1} and the union bound over all vertices. 
\end{proof}



%
%
%

\section{Applications to Different Graph Topologies}\label{sec:application}

\textbf{Cycles.} Recall that for the cycle, $V=\{0,\ldots,n-1\}$ is the set of vertices, and the distance between two vertices is $\dist(x,y) = \min\{y-x, x+n-y \}$ for any pair of vertices $x < y $.

The upper bound on the discrepancy follows directly from Theorem~\ref{thm:maintwo}, and it only remains to prove the lower bound. To this end, we will apply the lower bound in Theorem~\ref{thm:main1} and need to derive a lower bound on
$ \| \M_{.,u}^t - \mathbf{ \frac{1}{n}} \|_2^2$. Intuitively, if we had a simple random walk, we could immediately infer that this quantity is $\Omega( 1/\sqrt{t})$, since after $t$ steps, the random walk is with probability $\approx 1/\sqrt{t}$ at any vertex with distance at most $O(\sqrt{t})$. To prove that this also holds for the load balancing process, we first derive a concentration inequality that upper bounds the probability for the random walk to reach a distant state:

%
%
%
%

%

\begin{restatable}{lem}{upper}\label{lem:upper}
Consider the standard balancing circuit model on the cycle with round matrix $\M$. Then for any $u \in V$ and $\delta \in (0,n/2-1)$, we have
\begin{equation*}
\begin{aligned}
 \sum_{v \in V \colon \dist(u,v) \geq \delta} \M^t_{u,v}  \leq 2 \cdot \exp\left( - \frac{(\delta-2)^2}{8t}   \right).
\end{aligned}
\end{equation*}
\end{restatable}

\begin{proof}
The proof of the lemma above makes uses of the following variant of Azuma's concentration inequality for martingales, which can be for instance found in McDiarmid's survey on concentration inequalities.

\begin{lem}[{\cite[Theorem 3.13 \& Inequality 41]{McD98}}]\label{lem:mcdiarmid}
	Let $Z_1,Z_2,\ldots,Z_n$ be a martingale difference sequence with $a_k \leq Z_k \leq b_k$ for each $k$, for suitable constants $a_k, b_k$. Then for any $\delta \geq 0$,
	\begin{equation*}
	\begin{aligned}
	\Pro{  \max_{1 \leq j \leq n} \left| \sum_{i=1}^j Z_i \right| > \delta } &\leq 2 \cdot \exp\left(- \frac{2 \delta^2}{\sum_{k=1}^n (b_k-a_k)^2}  \right).
	\end{aligned}
	\end{equation*}
\end{lem}

Note that the balancing circuit on the cycle corresponds to the following random walk $(X_1,X_2,\ldots,X_{t})$ on the vertex set $V= \{-n/2+1,\ldots,0,\ldots,n/2-1\}$, where for any time-step $t \in \mathbb{N}$, $X_t$ denotes the position of the random walk after step $t$. First, we consider the transition for any odd $s$: If $X_s$ is odd, then with probability $1/2$, $X_{s+1} = X_{s}+1$ and otherwise $X_{s+1}=X_s$. If $X_s$ is even, then with probability $1/2$, $X_{s+1}=X_s-1$ and otherwise $X_{s+1}=X_s$ (additions and subtractions are under the implicit assumptions that $-n/2+1 \equiv n/2-1$ and $n/2 \equiv -n/2+1$). The case for even $s$ is analogous. 

We will couple the random walk $(X_t)_{t \geq 0}$ with another random walk $(Y_t)_{t \geq 0}$ on the integers $\mathbb{N}$, where again $Y_t$ denotes the position of the walk after step $t$. The transition probabilities are exactly the same as for the walk $(X_t)_{t \geq 0}$, the only difference is that we don't use the equivalences $-n/2+1 \equiv n/2-1$ and $n/2 \equiv -n/2+1$. It is clear that we can couple the transitions of the two walks so that they evolve identically as long as the walks do not reach any of the two boundary points $-n/2+1$ or $n/2-1$.

Let us first analyze $\Ex{Y_t}$ for an odd time step. As described above, the distribution of $Y_t-Y_{t-1}$ depends on whether $Y_{t-1}$ is even or not. However, notice regardless of where the random walk is at step $t-2$, the random walk will be at an odd or even vertex at step $t-1$ with probability $1/2$ each. Hence for any starting position $y$,
\begin{align*}
	\lefteqn{ \Ex{ Y_{t} - Y_{t-1} \, \mid \, Y_{0} = y } } \\
	&= \Pro{ Y_{t-1} \mbox{~even} } \cdot \left( \frac{1}{2} \cdot 1 + \frac{1}{2} \cdot 0 \right) + \Pro{ Y_{t-1} \mbox{~odd} } \cdot \left( \frac{1}{2} \cdot (-1) + \frac{1}{2} \cdot 0 \right) = 0,
\end{align*}
and further,
\begin{align*}
	| Y_{1} - Y_{0} | \leq 1.
\end{align*}
Combining the last two inequalities shows that for any start vertex $y$,
\begin{align*}
	\left| \Ex{ Y_{t} \, \mid \, Y_0 = y } - y \right| \leq 1.
\end{align*}
With the same arguments as before we conclude that for any fixed start vertex $Y_0=y_0$,
\begin{align*}
	\max_{a,b \in V} \left| \Ex{ Y_{t} - Y_{t-1} \, \mid \, Y_1=a } 
	- \Ex{ Y_{t} - Y_{t-1} \, \mid \, Y_1=b }  \right| \leq 2,
\end{align*}
because the expected differences of $Y_{t} - Y_{t-1}$ are all zero whenever $t \geq 3$.


Let us now consider the martingale $W_i = \Ex{ Y_t \, \mid \, Y_0,Y_1,\ldots,Y_i }$,
and let $Z_i := W_{i} - W_{i-1}$ be the corresponding martingale difference sequence. As shown before, $|W_{i}-W_{i-1} | \leq 2$. Hence by Lemma~\ref{lem:mcdiarmid},
\begin{align*}
	\Pro{ \max_{1 \leq j \leq t} \left| \sum_{i=1}^j Z_i \right|  \geq \delta } &\leq 2 \cdot \exp\left(  -\frac{ \delta^2 }{ 8 \cdot t  } \right)
\end{align*}

If for every $1 \leq j \leq t$, $ \sum_{i=1}^j W_i  < \delta$ holds, then this implies both random walks $(X_t)_{t \geq 0}$ and $(Y_t)_{t \geq 0}$ behave identically since none of them ever reaches any of the two boundary points $-n/2+1$ or $n/2-1$. In particular we conclude that for the original walk $(X_t)_{t \geq 0}$,
\begin{align*}
	\sum_{v \in V: \dist(u,v) \geq \delta} \M^t_{u,v} &=
	\Pro{ \left| \sum_{i=1}^t X_i \right|  \geq \delta } \\ &\leq
	\Pro{ \max_{1 \leq j \leq t} \left| \sum_{i=1}^j X_t \right|  \geq \delta } \\ &=
	\Pro{ \max_{1 \leq j \leq t} \left| \sum_{i=1}^j Y_t \right|  \geq \delta } 
	\\ &\leq
	\Pro{ \max_{1 \leq j \leq t} \left| \sum_{i=1}^j Z_t \right|  \geq \delta -2 } 
	\leq 2 \cdot \exp\left(  -\frac{ (\delta-2)^2 }{ 8 \cdot t   } \right),
\end{align*}
where the second-to-last inequality is due to the fact that $\Ex{  \left| \sum_{i=1}^j Y_t \right| } \leq 2$.
\end{proof}

With the help Lemma~\ref{lem:upper}, we can indeed verify our intuition:

\begin{restatable}{lem}{intuition}\label{lem:intuition}
Consider the standard balancing circuit model on the cycle with round matrix $\M$. Then for any vertex $u \in V$,
$ \| \M_{.,u}^t - \mathbf{ \frac{1}{n}} \|_2^2 = \Omega(1/\sqrt{t})$.
\end{restatable}
\begin{proof}

Define $S_{\delta} := \{ w \in V: \dist(w, u) \leq \delta \}$, so that $| S_\delta | = 2\delta$. With $\delta = 20\sqrt{t}$ and $t \geq 10$

\begin{align*}
\sum_{w \in S_{\delta}}\M^t_{w,u} & = 1 - \sum_{w \notin S_\delta} \M^t_{w,u} \geq 1 - 2\cdot \exp\left( -\frac{(\delta - 2)^2}{8t}\right)  \geq \frac{1}{2}.
\end{align*}
By Cauchy-Schwarz inequality,

\begin{equation*}
\begin{aligned}
 \| \M_{w,\cdot}^t \|_2^2 &\geq \sum_{w \in S_{\delta}} \left(\M^t_{w,u} \right)^2 & \geq \frac{1}{2\delta} \left( \sum_{w \in S{_\delta}}\M^t_{w,u} \right)^2 \geq \frac{1}{2\delta}\left(\frac{1}{2}\right)^2 = \Omega(t^{-1/2}).
\end{aligned}
\end{equation*}


\end{proof}

Lemma~\ref{lem:intuition} also proves that the factor $\sqrt{1/t}$ in the upper bound in Theorem~\ref{thm:markov} is best possible.
The lower bound on the discrepancy now follows by combining Lemma~\ref{lem:intuition} with Theorem~\ref{thm:main1} and Lemma~\ref{lem:twonodestoavg} stating that for any vertex $u \in V$, there exists another vertex $v \in V$ such that $ \| \M_{.,u}^t - \M_{.,v}^t \|_2^2 \geq \| \M_{.,u}^t - \mathbf{ \frac{1}{n}} \|_2^2 =\Omega(1/\sqrt{t})$.

\bigskip
\noindent \textbf{Tori.}
In this section we consider $r$-dimensional tori, where $r \geq 1$ is any constant. 
For the upper bound, note that the computation of $\M^t_{.,.}$ can be decomposed to independent computations in the $r$ dimensions, and each dimension has the same distribution as the cycle on $n^{1/r}$ vertices. Specifically, if we denote by $\widetilde{\M}$ the round matrix of the standard balancing circuit scheme on the cycle with $n^{1/r}$ vertices and $\M$ is the round matrix of the $r$-dimensional torus with $n$ vertices, then for any pair of vertices $x=(x_1,\ldots,x_r), v=(y_1,\ldots,y_r)$ on the torus we have
$
  \M_{x,y}^t = \prod_{i=1}^r \widetilde{\M}_{x_i,y_i}^t.
$
From Theorem~\ref{thm:markov}, $| \widetilde{\M}_{x_i,y_i}^t - \frac{1}{n^{1/r}} | = O(t^{-1/2})$, and therefore, since $r$ is constant,
\begin{align*}
  \M_{x,y}^t &\leq \prod_{i=1}^r \left( \frac{1}{n^{1/r}} + O(t^{-1/2}) \right)
  = O( t^{-r/2} + n^{-1}),
\end{align*}
and thus $\left\| \M_{x,y}^t - \mathbf{\frac{1}{n}} \right\|_2^2 = O(t^{-r/2})$ for any pair of vertices $x,y$. Hence by Lemma~\ref{lem:twonodestoavg}, $ \left \| \M_{.,u}^t - \M_{.,v}^t \right \|_{2}^2 = O( t^{-r/2})$. Plugging this bound into Theorem~\ref{thm:main1} yields that the load difference between any pair of the nodes $u$ and $v$ at round $t$ is at most
$ O( t^{-r/4} \cdot \sigma \cdot \log^{3/2} n  + \sqrt{\log n})$ with probability at least $1-2 n^{-2}$. The bound on the discrepancy now simply follows by the union bound.

We now turn to the lower bound on the discrepancy. With the same derivation as in Lemma~\ref{lem:intuition} we obtain the following result:
\begin{restatable}{lem}{intuitiontwo}\label{lem:intuitiontwo}
Consider the standard balancing circuit model on the $r$-dimensional torus with round matrix $\M$. Then for any vertex $u \in V$,
$ \| \M_{.,u}^t - \mathbf{ \frac{1}{n}} \|_2^2 = \Omega(t^{-r/2})$.
\end{restatable}
As before, the lower bound on the torus now follows by combining Lemma~\ref{lem:intuitiontwo} with the general lower bound given in Theorem~\ref{thm:main1}. 

%

%

\bigskip

\noindent\textbf{Expanders.}
The upper bound $O(\lambda(\M)^{t/4} \cdot \sigma \cdot (\log n)^{3/2} + \sqrt{\log n})$ for expanders follows immediately from Theorem~\ref{thm:main3}.
%
%
For the lower bound, since the round matrix consists of $d$ matchings, it is easy to verify that whenever $\M_{u,v}^{t} > 0$, we have $\M_{u,v}^t \geq 2^{-d \cdot t}$. Consequently, for any vertex $u \in V$, $\left\| \M_{.,u}^t - \mathbf{\frac{1}{n}} \right\|_2^2
=\Omega( 2^{-d \cdot t})$. Plugging this into Theorem~\ref{thm:main1} yields a lower bound on the discrepancy which is $\Omega( 2^{-d \cdot t/2} \cdot \sigma / \sqrt{\log \sigma})$.

\bigskip
\noindent \textbf{Hypercubes.}
For the hypercube, there is a worst-case bound of $\log_2 \log_2 n + O(1)$ \cite[Theorem 5.1 $\&$ 5.3]{mavronicolas2010impact} for any input after $\log_2 n$ iterations of the dimension-exchange, i.e., after one execution of the round matrix. Hence, we will only analyze the discrepancy after $s$ matchings, where $1 \leq s < \log_2 n$.

The derivation of the lower bound is almost analogous to the one for expanders, since for any pair of vertices $u,v$, $\prod_{i=s}^t \M_{u,v}^{(s)} \in \{0,2^{-t}\}$ (recall that $\M_{.,.}^{(s)}$ is the matching applied in the $s$-step of the dimension exchange). The only difference is that we are counting matchings individually and not full periods. By applying the same analysis as in Theorem~\ref{thm:main3}, but with the stronger inequality $| \prod_{s=1}^t \M_{u,v}^{(s)} - \frac{1}{n} | \leq 2^{-t}$, and we obtain that the upper bound of the discrepancy is $O(2^{-t/2} \cdot \sigma \cdot (\log n)^{3/2} + \sqrt{\log n})$. Applying Theorem~\ref{thm:main1}, we obtain the lower bound $\Omega( 2^{-t/2} \cdot \sigma / \sqrt{\log \sigma} )$.

\section{Discussion and Empirical Results}\label{sec:discussion}

\subsection{Average-Case versus Worst-Case}

We will now compare our average-case to a worst-case scenario on cycles, 2D-tori and hypercubes. For the sake of concreteness, we always assume that the input is drawn from the uniform distribution $\mathsf{Uni}[0,2K]$, where $K$ will be specified later. Note that the total number of tokens is $\approx n \cdot K$, and the initial discrepancy will be $\Theta(K)$. Our choice for the worst-case load vector will have the same number of tokens and initial discrepancy, however, the exact definition of the vector as well as the choice of the parameter $K$ will depend on the underlying topology. 

\textbf{Cycles.}
As one representative of a worst-case setting, fix an arbitrary node $u \in V$ and let all nodes with distance at most $n/4$ initially have a load of $2K$ while all other nodes have load $0$. This gives rise to a load vector with $n \cdot K$ tokens and initial discrepancy $2K$.

\textbf{2D-Tori.} Again, we fix an arbitrary node $u \in V$ and assign a load of $2K$ to the $n/2$-nearest neighbors of $u$ and load $0$ to the other nodes. Again, this defines a load vector with $n \cdot K$ tokens and initial discrepancy $2K$.

The next result provides a lower bound on the discrepancy for cycles and 2D-tori in the aforementioned worst-case setting. It essentially shows that for worst-case inputs, $\Omega(n^2)$ rounds and $\Omega(n)$ rounds are necessary for the cycle, 2D-tori, respectively, in order to reduce the discrepancy by more than a constant factor. This stands in sharp contrast to Theorem~\ref{thm:maintwo}, proving a decay of the discrepancy by $\approx t^{-1/4}$, starting from the first round.

\begin{restatable}{proposition}{worstcase} \label{pro:worstcase}
For the aforementioned worst-case setting on the cycle, it holds for any round $t > 0$ that
$
\disc(x^{(t)}) \geq \frac{1}{8} \cdot K \cdot \left(1 - \exp \left( - \frac{n^2}{2048 t} \right) \right) - \sqrt{48 \log n},
$
with probability at least $1-n^{-1}$.
Further, for 2D-tori, it holds for any round $t > 0$ that
$
\disc(x^{(t)}) \geq \frac{1}{8} \cdot K \cdot \left(1 - \exp \left( - \frac{n}{2048 t} \right) \right) - \sqrt{48 \log n},
$
with probability at least $1-n^{-1}$.
\end{restatable}

\begin{proof}
	We first consider the case of a cycle.
	Let $S_1$ be the subset of nodes that have a non-zero initial load; so $|S_1|=n/2$. Clearly, there is a subset of nodes $S_2 \subseteq V$ with $|S_2| = n/8$ so that for each node $u \in S_2$, only nodes $v$ with $\dist(u,v) \geq n/16$ can have $x_v^{(0)} > 0$.
	
	We will now derive a lower bound on the discrepancy in this worst-case setting by upper bounding the load of vertices in the subset $S_2$. To lower bound the discrepancy at round $t$, recall that by Lemma~\ref{lem:upper} we have that 
	\begin{align*}
		\sum_{v \in V \colon \dist(u,v) \geq \delta} \M^{t}_{u,v} \leq 2 \cdot \exp \left( - \frac{(\delta-2)^2}{8t} \right).
	\end{align*}
	Let us now choose $\delta=n/16$, and we thus conclude that
	\begin{align*}
		\sum_{u \in S_1} \sum_{v \in S_2} \M^{t}_{u,v} &\leq
		\sum_{u \in S_1} \sum_{v \in V \colon \dist(u,v) \geq \delta} \M^{t}_{u,v} \leq 2 \cdot |S_1| \cdot \exp \left( - \frac{n^2}{2048 t} \right).
	\end{align*}
	This implies for the total load of vertices in $S_2$ at time $t$:
	\begin{align*}
		\sum_{v \in S_2} \xi_v^{(t)} &= \sum_{v \in S_2} \sum_{u \in S_1} \xi_u^{(0)} \cdot \M_{u,v}^{(t)} \\
		&= 2 K \sum_{u \in S_1}  \sum_{v \in S_2} \M^{(t)}_{u,v} \\
		&\leq K \cdot n \cdot \exp \left( - \frac{n^2}{2048 t} \right),
	\end{align*}
	where $K$ is the average load.
	Recalling that $|S_2|=n/8$, by the pigeonhole principle there exists a node $v \in S_2$ such that
	\begin{align*}
		\xi_v^{(t)} &\geq \frac{1}{|S_2|} \cdot K \cdot n \cdot \exp \left( -  \frac{n^2}{2048 t} \right) \geq \frac{1}{8} \cdot K \cdot \exp \left( - \frac{n^2}{2048 t} \right).
	\end{align*} 
	
	%
	%
	%
	This immediately implies the following lower bound on the discrepancy:
	\begin{align*}
		\disc(\xi^{(t)}) &\geq \bar{\xi} - \xi_v^{(t)} \geq \bar{\xi} \cdot \left(1 - \exp \left( - \frac{n^2}{2048 t} \right) \right),
	\end{align*}
	where $\bar{\xi} = K$ is the average load.
	The corresponding lower bound on $\disc(x^{(t)})$ follows by Theorem~\ref{thm:errorbound} and the union bound.
	
	The proof for the 2-dimensional torus is almost identical. Again, let $S_1$ be the set of nodes that have a non-zero load. Clearly, there is a subset $S_2 \subseteq V$ with $|S_2|=n/8$ so that  for each node $u \in S_2$, only nodes $v$ with $\dist(u,v) \geq \sqrt{n}/16$ can have $x_v^{(0)} > 0$.
	
	Let us now view $\M$ as the transition matrix of a Markov chain. Then $\M^{t}$ is obtained by running two independent Markov chains  (one for each dimension), where each of the two Markov  chains corresponds to the round matrix of the cycle. We can still apply Lemma~\ref{lem:upper} as before, even though here the size of each cycle is $\sqrt{n}$, to obtain that
	\begin{align*}
		\sum_{v \in V \colon \dist(u,v) \geq \delta} \M^{t}_{u,v} \leq 2 \cdot \exp \left( - \frac{(\delta-2)^2}{8t} \right).
	\end{align*}
	Here we choose $\delta=\sqrt{n}/16$, and the remaining part of the proof is exactly the same as before.
\end{proof}



%
%
%

\textbf{Hypercube.} Regarding the hypercube, we will consider only $\log_2 n$ rounds, since the discrepancy is $\log \log_2 n +O(1)$ after $\log_2 n$ rounds and $O(1)$ after $2 \log_2 n$ rounds \cite{mavronicolas2010impact}. A natural corresponding worst-case distribution is to have load $2 K$ on all nodes whose $\log_2 n$-th bit is equal to one and load $0$ otherwise. This way, the discrepancy is only reduced in the final round $\log_2 n$.

\subsection{Experimental Setup}

For each of the three graphs cycles, 2D-tori and hypercube, we consider two comparative experiments with an average-case load vector and a worst-case initial load vector each. The plots and tables on the next two pages display the results, where for each case we took the average discrepancy over 10 independent runs. 

The first experiment considers a ``lightly loaded case'', where the theoretical results suggest that a small (i.e., constant or logarithmic) discrepancy is reached well before the expected ``worst-case load balancing times'', which are $\approx n^2$ for cycles and $\approx n$ for 2D-tori. The second experiments considers a ``heavily loaded case'', where the theoretical results suggest that a small discrepancy is not reached faster than in the worst-case. 

Specifically, for cycles and 2D-tori, we choose for the lightly loaded case $K=\sqrt{n}$ and for the heavily loaded case $K=n^2$. The experiments confirm the theoretical results in the sense that for both choices of $K$, we have a much quicker convergence of the discrepancy than in the corresponding worst cases. However, the experiments also demonstrate that only in the lightly loaded case we reach a small discrepancy quickly, whereas in the heavily loaded case there is no big difference between worst-case and average-case if it comes to the time to reach a small discrepancy.

On the hypercube, since we are interested in the case where $1 \leq t \leq \log_2 n$, our bounds on the discrepancy indicates that we should choose $K$ smaller than in the case of cycles and 2D-tori. That is why we choose $K=n^{1/4}$ in the lightly loaded case and $K=n$ in the heavily loaded case (As a side remark, we note that due to the symmetry of the hypercube, any initial load vector sampled from $\Uni[0,\beta \cdot (n-1)]$ is equivalent to an initial load vector sampled from $\Uni[0, n-1]$.) With these adjustments of $K$ in both cases, the experimental results of the hypercube are inline with the ones for the cycle and 2D-tori. 

The details of the experiments containing plots and tables with the sampled discrepancies can be found on the following two pages (Section~\ref{sec:experiments}).

\newpage

\appendix

\section{Experimental Data and Charts}
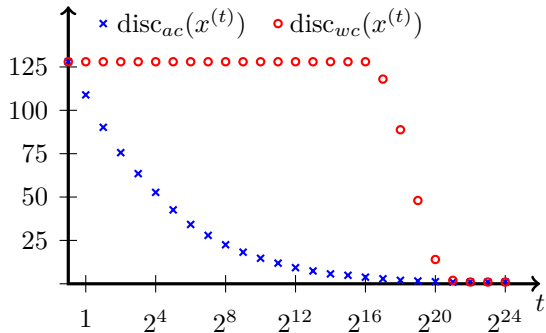
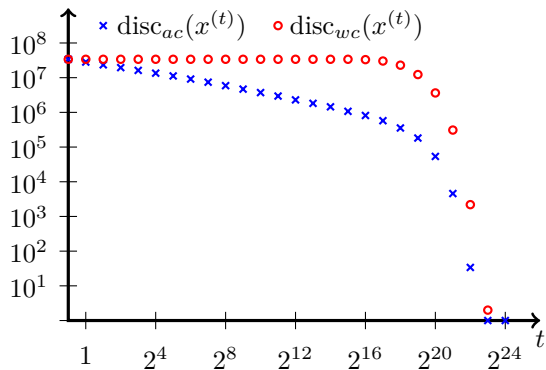
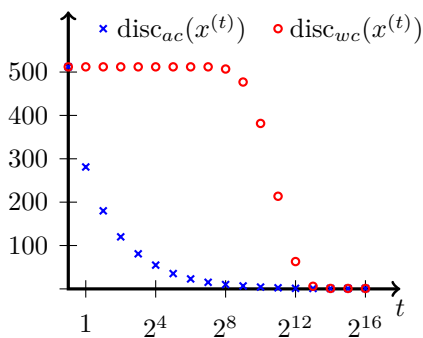
\label{sec:experiments}

		\begin{tikzpicture}[scale=0.23]
		\draw [very thick][<->] (0,16)  -- (0,0) -- (27,0) node[below]{$t$};	
		\node [left=4pt] at (0,2.5) {$25$};
		\node [left=4pt] at (0,5.0) {$50$};
		\node [left=4pt] at (0,7.5) {$75$};
		\node [left=4pt] at (0,10) {$100$};
		\node [left=4pt] at (0,12.5) {$125$};
		\node [below] at (1,-1) {1};
		\node [below] at (5,-1) {$2^{4}$};
		\node [below] at (9,-1) {$2^{8}$};
		\node [below] at (13,-1) {$2^{12}$};
		\node [below] at (17,-1) {$2^{16}$};
		\node [below] at (21,-1) {$2^{20}$};
		\node [below] at (25,-1) {$2^{24}$};
		\foreach \x in {1,5,...,25}
		{
		    \draw (\x,-0.5) -- (\x,+0.5);
		}
		\foreach \y in {0,2.5,...,12.5}
		{
		    \draw (-0.5,\y) -- (+0.5,\y);
		}
		
		\draw [color=blue,thick] -- plot[smooth,mark=x, mark size=8] coordinates {(0,12.8)(1,10.89)(2,9.02)(3,7.56)(4,6.35)(5,5.27)(6,4.26)(7,3.42)(8,2.79)(9,2.25)(10,1.82)(11,1.47)(12,1.19)(13,0.93)(14,0.74)(15,0.57)(16,0.49)(17,0.38)(18,0.29)(19,0.2)(20,0.16)(21,0.12)(22,0.1)(23,0.1)(24,0.1)(25,0.1)};

		\draw [color=red,thick] -- plot[smooth,mark=o, mark size=6] coordinates {(0,12.80)(1,12.80)(2,12.8)(3,12.8)(4,12.80)(5,12.80)(6,12.8)(7,12.8)(8,12.8)(9,12.8)(10,12.8)(11,12.80)(12,12.8)(13,12.8)(14,12.8)(15,12.8)(16,12.8)(17,12.8)(18,11.8)(19,8.88)(20,4.8)(21,1.4)(22,0.2)(23,0.1)(24,0.1)(25,0.1)};		
		
		\draw [thick,color=red] --plot[smooth,mark=o,mark size=6] coordinates {(12,15)};
		\node () at (16.5,15) {$\disc_{wc}(x^{(t)})$};
		\node () at (6.5,15) {$\disc_{ac}(x^{(t)})$};
		\draw [thick,color=blue] --plot[smooth,mark=x,mark size=8] coordinates {(2,15)};
		
		\node () at (37,3.2)  
		{
		\small{
				\begin{minipage}[t]{0.4\textwidth}
		\centering
	\begin{tabular}{|c|c|c|}\hline
			$t$ & $\disc_{ac}$ & $\disc_{wc}$ \\\hline 
			0 & 128.0 & 128.0 \\  $2^0$ & 108.9 & 128.0\\ $2^2$ & 75.6 & 128.0\\ $2^4$ & 52.7 & 128.0\\ $2^6$ & 34.2 & 128.0\\ $2^8$ & 22.5 & 128.0\\ $2^{10}$ & 14.7 & 128.0\\ $2^{12}$ & 9.3 & 128.0\\ $2^{14}$ & 5.7 & 128.0\\ $2^{16}$ & 3.8 & 128.0\\ $2^{18}$ & 2.0 & 88.0\\ $2^{20}$ & 1.2 & 14.0\\ $2^{22}$ & 1.0 & 1.0 \\ $2^{24}$ & 1.0 & 1.0 \\
			\hline
		\end{tabular}
		\end{minipage}
		}
		};

\begin{scope}[yshift=-31cm]
		\draw [very thick][<->] (0,18) -- (0,0) -- (27,0) node[below]{$t$};	
		\node [left=4pt] at (0,2) {$10^1$};
		\node [left=4pt] at (0,4) {$10^2$};
		\node [left=4pt] at (0,6) {$10^3$};
		\node [left=4pt] at (0,8) {$10^4$};
		\node [left=4pt] at (0,10) {$10^5$};
		\node [left=4pt] at (0,12) {$10^6$};
		\node [left=4pt] at (0,14) {$10^7$};
			\node [left=4pt] at (0,16) {$10^8$};
		\node [below] at (1,-1) {1};
		\node [below] at (5,-1) {$2^{4}$};
		\node [below] at (9,-1) {$2^{8}$};
		\node [below] at (13,-1) {$2^{12}$};
		\node [below] at (17,-1) {$2^{16}$};
		\node [below] at (21,-1) {$2^{20}$};
		\node [below] at (25,-1) {$2^{24}$};
		\foreach \x in {1,5,...,25}
		{
		    \draw (\x,-0.5) -- (\x,+0.5);
		}
		\foreach \y in {0,2,...,16}
		{
		    \draw (-0.5,\y) -- (+0.5,\y);
		}
		
		\draw [thick, color=blue] -- plot[smooth,mark=x, mark size=8]  coordinates {(0,15.06)(1,14.9)(2,14.74)(3,14.58)(4,14.42)(5,14.26)(6,14.1)(7,13.92)(8,13.74)(9,13.54)(10,13.34)(11,13.14)(12,12.94)(13,12.72)(14,12.52)(15,12.32)(16,12.06)(17,11.82)(18,11.52)(19,11.1)(20,10.52)(21,9.46)(22,7.32)(23,3.06)(24,0)(25,0)};

		\draw [thick, color=red] -- plot[smooth,mark=o, mark size=6] coordinates {(0,15.06)(1,15.06)(2,15.06)(3,15.06)(4,15.06)(5,15.06)(6,15.06)(7,15.06)(8,15.06)(9,15.06)(10,15.06)(11,15.06)(12,15.06)(13,15.06)(14,15.06)(15,15.06)(16,15.06)(17,15.04)(18,14.96)(19,14.72)(20,14.18)(21,13.12)(22,10.98)(23,6.68)(24,0.60)};		
		
		\draw [thick, color=red] --plot[smooth,mark=o,mark size=6] coordinates {(12,17)};
		\node () at (16.5,17) {$\disc_{wc}(x^{(t)})$};
		\node () at (6.5,17) {$\disc_{ac}(x^{(t)})$};
		\draw [thick, color=blue] --plot[smooth,mark=x,mark size=8] coordinates {(2,17)};
		
		\node () at (38.5,6)  
		{
				\begin{minipage}[t]{0.4\textwidth}
		\centering
		\small{
		\begin{tabular}{|c|c|c|}\hline
			$t$ & $\disc_{ac}$ & $\disc_{wc}$ \\\hline 
			0 & $3.35 \times 10^7$ & $3.35 \times 10^7$ \\  $2^0$ & $2.83 \times 10^7$ & $3.35 \times 10^7$\\ $2^2$ & $1.96 \times 10^7$ & $3.35 \times 10^7$\\ $2^4$ & $1.34 \times 10^7$ & $3.35 \times 10^7$\\ $2^6$ & $9.17 \times 10^6$ & $3.35 \times 10^7$\\ $2^8$ & $5.94 \times 10^6$ & $3.35 \times 10^7$\\ $2^{10}$ & $3.72 \times 10^6$ & $3.35 \times 10^7$\\ $2^{12}$ & $2.30 \times 10^6$ & $3.35 \times 10^7$\\ $2^{14}$ & $1.43 \times 10^6$ & $3.35 \times 10^7$\\ $2^{16}$ & $8.19 \times 10^5$ & $3.32 \times 10^7$\\ $2^{18}$ & $3.58 \times 10^5$ & $2.30 \times 10^7$\\ $2^{20}$ & $5.34 \times 10^4$ & $3.62 \times 10^6$\\ $2^{22}$ & $3.35 \times 10^1$ & $2.21 \times 10^3$ \\ $2^{24}$ & 1.0 & 1.0 \\
			\hline
		\end{tabular}
		}
		\end{minipage}
		};	
		\end{scope}
		
\begin{scope}[yshift=-58cm]
		\draw [very thick][<->] (0,16)  -- (0,0) -- (19,0) node[below]{$t$};		
	\node [left=4pt] at (0,2.5) {$100$};
		\node [left=4pt] at (0,5.0) {$200$};
		\node [left=4pt] at (0,7.5) {$300$};
		\node [left=4pt] at (0,10) {$400$};
		\node [left=4pt] at (0,12.5) {$500$};
		\node [below] at (1,-1) {1};
		\node [below] at (5,-1) {$2^{4}$};
		\node [below] at (9,-1) {$2^{8}$};
		\node [below] at (13,-1) {$2^{12}$};
		\node [below] at (17,-1) {$2^{16}$};
		\foreach \x in {1,5,...,17}
		{
		    \draw (\x,-0.5) -- (\x,+0.5);
		}
		\foreach \y in {0,2.5,...,12.5}
		{
		    \draw (-0.5,\y) -- (+0.5,\y);
		}

		\draw [thick, color=blue] -- plot[smooth,mark=x, mark size=8] coordinates {(0,12.8)(1,7.0275)(2,4.5)(3,3)(4,2.0225)(5,1.37)(6,0.8825)(7,0.5750)(8,0.38)(9,0.25)(10,0.17)(11,0.105)(12,0.0625)(13,0.0325)(14,0.025)(15,0.025)(16,0.025)(17,0.025)};

		\draw [thick, color=red] -- plot[smooth,mark=o, mark size=6] coordinates {(0,12.8)(1,12.8)(2,12.8)(3,12.8)(4,12.8)(5,12.8)(6,12.8)(7,12.8)(8,12.8)(9,12.675)(10,11.925)(11,9.5375)(12,5.34)(13,1.575)(14,0.15)(15,0.025)(16,0.025)(17,0.025)};		
		
		\draw [thick, color=red] --plot[smooth,mark=o,mark size=6] coordinates {(12,15)};
		\node () at (16.5,15) {$\disc_{wc}(x^{(t)})$};
		\node () at (6.5,15) {$\disc_{ac}(x^{(t)})$};
		\draw [thick, color=blue] --plot[smooth,mark=x,mark size=8] coordinates {(2,15)};
		
		\node () at (38.5,6)  
		{
				\begin{minipage}[t]{0.4\textwidth}
		\centering
		\small{
	\begin{tabular}{|c|c|c|}\hline
		$t$ & $\disc_{ac}$ & $\disc_{wc}$ \\\hline 
		0 & 512.0 & 512.0 \\  $2^0$ & 281.1 & 512.0 \\  $2^2$ & 120.0 & 512.0 \\  $2^4$ & 54.8 & 512.0 \\  $2^6$ & 23.0 & 512.0 \\    $2^8$ & 10.0 & 507.0 \\ $2^{10}$ & 4.2 & 381.5 \\ $2^{12}$ & 1.3 & 63.0 \\ $2^{14}$ & 1.0 & 1.0 \\ $2^{16}$ & 1.0 & 1.0 \\
		\hline
		\end{tabular}
		}
		\end{minipage}
		};	
		
		\node () at (24,-11)
		{
		\begin{minipage}{0.9\textwidth}
		\textbf{Experimental Results:} Experiments (i) on the cycle with $n=2^{12}$ and initial discrepancy $2^7=128$, (ii) on the cycle with $n=2^{12}$ and initial discrepancy $2^{25}=33,554,432$, and (iii) on the 2D-torus with $n=2^{16}$ and initial discrepancy of \textbf{$2^{9}$}. For the heavily loaded case, we used logarithmic scaling on the $y$-axis to highlight the behavior when $t$ is close to the worst-case load balancing time.
		\end{minipage}
		};
		
		\end{scope}
		
		\end{tikzpicture}

			
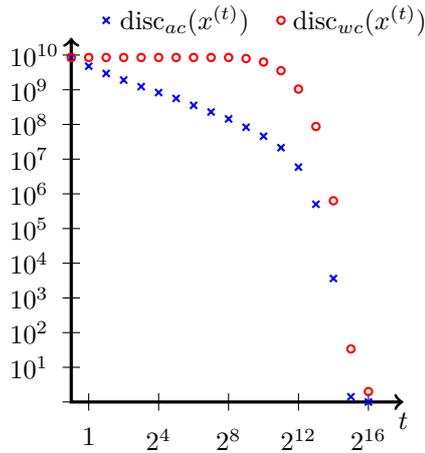
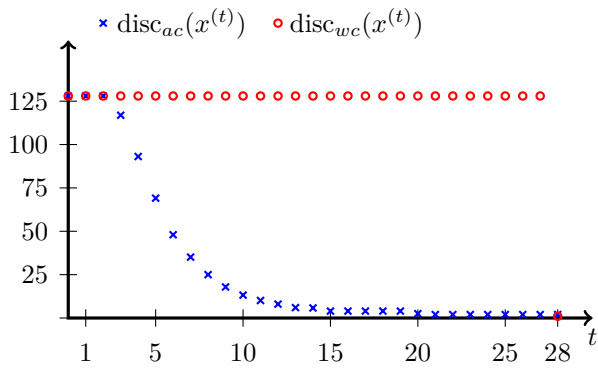
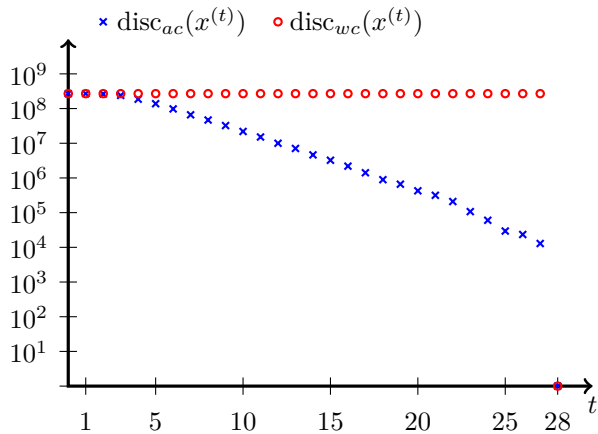
\begin{tikzpicture}[scale=0.23]
		\draw [very thick][<->] (0,21) -- (0,0) -- (19,0) node[below]{$t$};	
		\node [left=5pt] at (0,2) {$10^1$};
		\node [left=5pt] at (0,4) {$10^2$};
		\node [left=5pt] at (0,6) {$10^3$};
		\node [left=5pt] at (0,8) {$10^4$};
		\node [left=5pt] at (0,10) {$10^5$};
		\node [left=5pt] at (0,12) {$10^6$};
			\node [left=5pt] at (0,14) {$10^7$};
		\node [left=5pt] at (0,16) {$10^8$};
		\node [left=5pt] at (0,18) {$10^9$};
		\node [left=1pt] at (0,20) {$10^{10}$};
			\node [below] at (1,-1) {1};
		\node [below] at (5,-1) {$2^{4}$};
		\node [below] at (9,-1) {$2^{8}$};
		\node [below] at (13,-1) {$2^{12}$};
		\node [below] at (17,-1) {$2^{16}$};
		\foreach \x in {1,5,...,19}
		{
		    \draw (\x,-0.5) -- (\x,+0.5);
		}
		\foreach \y in {0,2,...,20}
		{
		    \draw (-0.5,\y) -- (+0.5,\y);
		}
		


		\draw [color=blue,thick] -- plot[smooth,mark=x, mark size=8] coordinates {(0,19.86)(1,19.36)(2,18.94)(3,18.56)(4,18.18)(5,17.84)(6,17.5)(7,17.1)(8,16.72)(9,16.32)(10,15.84)(11,15.32)(12,14.66)(13,13.54)(14,11.40)(15,7.12)(16,0.3)(17,0)};
		
		\draw [color=red,thick] -- plot[smooth,mark=o, mark size=6] coordinates {(0,19.86)(1,19.86)(2,19.86)(3,19.86)(4,19.86)(5,19.86)(6,19.86)(7,19.86)(8,19.86)(9,19.86)(10,19.80)(11,19.60)(12,19.1)(13,18.04)(14,15.88)(15,11.60)(16,3.06)(17,0.6)};

		\draw [thick,color=red] --plot[smooth,mark=o,mark size=6] coordinates {(12,22)};
		\node () at (16.5,22) {$\disc_{wc}(x^{(t)})$};
		\node () at (6.5,22) {$\disc_{ac}(x^{(t)})$};
		\draw [thick,color=blue] --plot[smooth,mark=x,mark size=8] coordinates {(2,22)};
		
		\node () at (37,6.2)  
		{
		\small{
				\begin{minipage}[t]{0.4\textwidth}
		\centering
		\begin{tabular}{|c|c|c|}\hline
		$t$ & $\disc_{ac}$ & $\disc_{wc}$ \\\hline 
		0 & $8.59 \times 10^9$ & $8.59 \times 10^9$ \\  $2^0$ & $4.83 \times 10^9$ & $8.59 \times 10^9$ \\  $2^2$ & $1.89 \times 10^9$ & $8.59 \times 10^9$ \\  $2^4$ & $8.28 \times 10^8$ & $8.59 \times 10^9$ \\  $2^6$ & $3.58 \times 10^8$ & $8.59 \times 10^9$ \\    $2^8$ & $1.44 \times 10^8$ & $8.50 \times 10^9$ \\ $2^{10}$ & $4.52 \times 10^7$ & $6.38 \times 10^9$ \\ $2^{12}$ & $5.91 \times 10^6$ & $1.03 \times 10^9$ \\ $2^{14}$ & $3.59 \times 10^3$ & $6.33 \times 10^5$ \\ $2^{16}$ & 1.0 & 2.0\\
		\hline
		\end{tabular}
		\end{minipage}
		}
		};

\begin{scope}[yshift=-26cm]
				\draw [very thick][<->] (0,16) -- (0,0) -- (30,0) node[below]{$t$};	
		\node [left=4pt] at (0,2.5) {$25$};
		\node [left=4pt] at (0,5.0) {$50$};
		\node [left=4pt] at (0,7.5) {$75$};
		\node [left=4pt] at (0,10) {$100$};
		\node [left=4pt] at (0,12.5) {$125$};
		\node [below] at (1,-1) {1};
		\node [below] at (5,-1) {$5$};
		\node [below] at (10,-1) {$10$};
		\node [below] at (15,-1) {$15$};
		\node [below] at (20,-1) {$20$};
		\node [below] at (25,-1) {$25$};
		\node [below] at (28,-1) {$28$};
			\foreach \x in {1,5,10,15,20,25,28}
		{
		    \draw (\x,-0.5) -- (\x,+0.5);
		}
		\foreach \y in {0,2.5,...,12.5}
		{
		    \draw (-0.5,\y) -- (+0.5,\y);
		}

		\draw  [thick, color=blue] -- plot[smooth,mark=x, mark size=8]  coordinates {(0,12.8) (1,12.8) (2,12.8) (3,11.69) (4,9.31) (5,6.91) (6,4.8) (7,3.51) (8,2.5) (9,1.79) (10,1.32) (11,1.01) (12,0.8) (13,0.6) (14,0.58) (15,0.4) (16,0.4) (17,0.4) (18,0.4) (19,0.4) (20,0.25) (21,0.2) (22,0.2) (23,0.2) (24,0.2) (25,0.2) (26,0.2) (27,0.2) (28,0.2)};
		
		
		\draw  [thick, color=red] -- plot[smooth,mark=o, mark size=6]  coordinates {(0,12.8) (1,12.8) (2,12.8) (3,12.8) (4,12.8) (5,12.8) (6,12.8) (7,12.8) (8,12.8) (9,12.8) (10,12.8) (11,12.8) (12,12.8) (13,12.8) (14,12.8) (15,12.8) (16,12.8) (17,12.8) (18,12.8) (19,12.8) (20,12.8) (21,12.8) (22,12.8) (23,12.8) (24,12.8) (25,12.8) (26,12.8) (27,12.8) (28,0.1)};

		\draw [thick, color=red] --plot[smooth,mark=o,mark size=6] coordinates {(12,17)};
		\node () at (16.5,17) {$\disc_{wc}(x^{(t)})$};
		\node () at (6.5,17) {$\disc_{ac}(x^{(t)})$};
		\draw [thick, color=blue] --plot[smooth,mark=x,mark size=8] coordinates {(2,17)};
		
		\node () at (39.5,6)  
		{
				\begin{minipage}[t]{0.4\textwidth}
		\centering
		\small{
	\begin{tabular}{|c|c|c|}\hline
			$t$  & $\disc_{ac}$ & $\disc_{wc}$ \\\hline 
			  0 & 128.0 & 128.0 \\
			1 & 128.0 & 128.0 \\  
			2  & 128.0 & 128.0  \\
			3 & 116.9  & 128.0\\
			4 &  93.1 & 128.0\\ 
			 5 & 69.1 & 128.0\\     10 & 13.2 & 128.0\\    15 & 5.8 & 128.0\\
			20 & 2.5 & 128.0\\    25 & 2.0 & 128.0 \\ 
			27 & 2.0 & 128.0 \\ 
			28 & 2.0 & 2.0 \\
			\hline
		\end{tabular}
		}
		\end{minipage}
		};	
		\end{scope}
		
\begin{scope}[yshift=-58cm]
		\draw [very thick][<->] (0,20) -- (0,0) -- (30,0) node[below]{$t$};	
		\node [left=4pt] at (0,2) {$10^1$};
		\node [left=4pt] at (0,4) {$10^2$};
		\node [left=4pt] at (0,6) {$10^3$};
		\node [left=4pt] at (0,8) {$10^4$};
		\node [left=4pt] at (0,10) {$10^5$};
		\node [left=4pt] at (0,12) {$10^6$};
		\node [left=4pt] at (0,14) {$10^7$};
			\node [left=4pt] at (0,16) {$10^8$};
			\node [left=4pt] at (0,18) {$10^9$};
		\node [below] at (1,-1) {1};
		\node [below] at (5,-1) {$5$};
		\node [below] at (10,-1) {$10$};
		\node [below] at (15,-1) {$15$};
		\node [below] at (20,-1) {$20$};
		\node [below] at (25,-1) {$25$};
		\node [below] at (28,-1) {$28$};

	\foreach \x in {1,5,10,15,20,25,28}
		{
		    \draw (\x,-0.5) -- (\x,+0.5);
		}
		\foreach \y in {0,2,...,18}
		{
		    \draw (-0.5,\y) -- (+0.5,\y);
		}

		\draw  [thick, color=blue] -- plot[smooth,mark=x, mark size=8]  coordinates {(0,16.86) (1,16.86) (2,16.84) (3,16.76) (4,16.54) (5,16.28) (6,15.98) (7,15.64) (8,15.34) (9,15.02) (10,14.68) (11,14.36) (12,14.00) (13,13.7) (14,13.33) (15,13.02) (16,12.68) (17,12.3) (18,11.9) (19,11.64) (20,11.26) (21,11) (22,10.64) (23,10.06) (24,9.56) (25,8.94) (26,8.74) (27,8.22) (28,0)};
		
		\draw [thick, color=red] -- plot[smooth,mark=o, mark size=6] coordinates {(0,16.86) (1,16.86) (2,16.86) (3,16.86) (4,16.86) (5,16.86) (6,16.86) (7,16.86) (8,16.86) (9,16.86) (10,16.86) (11,16.86) (12,16.86) (13,16.86) (14,16.86) (15,16.86) (16,16.86) (17,16.86) (18,16.86) (19,16.86) (20,16.86) (21,16.86) (22,16.86) (23,16.86) (24,16.86) (25,16.86) (26,16.86) (27,16.86) (28,0)};
		
		\draw [thick, color=red] --plot[smooth,mark=o,mark size=6] coordinates {(12,21)};
		\node () at (16.5,21) {$\disc_{wc}(x^{(t)})$};
		\node () at (6.5,21) {$\disc_{ac}(x^{(t)})$};
		\draw [thick, color=blue] --plot[smooth,mark=x,mark size=8] coordinates {(2,21)};
		
		\node () at (43,10)  
		{
				\begin{minipage}[t]{0.4\textwidth}
		\centering
		\small{
		\begin{tabular}{|c|c|c|}\hline
			$t$ & $\disc_{ac}$ & $\disc_{wc}$ \\\hline 
			0 & $2.68 \times 10^8$ & $2.68 \times 10^8$ \\  
			1 & $2.68 \times 10^8$ & $2.68 \times 10^8$\\  
			2  & $2.65 \times 10^8$    & $2.68 \times 10^8$  \\
			3 & $2.41 \times 10^8$ & $2.68 \times 10^8$\\
			4 & $1.88 \times 10^8$ & $2.68 \times 10^8$\\ 
			  5 & $1.37 \times 10^8$ & $2.68 \times 10^8$\\     10  & $2.21 \times 10^7$ & $2.68 \times 10^8$\\    15 & $3.26 \times 10^6$ & $2.68 \times 10^8$\\
			20 & $4.23 \times 10^5$ & $2.68 \times 10^8$\\    25 & $2.98 \times 10^4$ & $2.68 \times 10^8$ \\
			27 & $1.29 \times 10^4$  &  $2.68 \times 10^8$ \\ 
			28 & $3.0$ & $3.0$ \\ 
			\hline
		\end{tabular}
		}
		\end{minipage}
		};	
		
		\node () at (24,-9)
		{
		\begin{minipage}{0.9\textwidth}
		\textbf{Experimental Results (cntd.):} Experiments (iv) on the 2D-torus with $n=2^{16}$ and initial discrepancy $2^{33}=8,589,934,592$, (v) on the hypercube with $n=2^{28}$ and initial discrepancy $256$, and (vi) on the hypercube with $n=2^{28}$ and initial discrepancy of $2^{28}=268,435,456$. For the heavily loaded cases, we used logarithmic scaling on the $y$-axis to highlight the behaviour when $t$ is close to the worst-case load balancing time.
		\end{minipage}
		};
		
		\end{scope}
		
		\end{tikzpicture}

\section{Concentration Tools}

\begin{lem}[{\cite[Theorem A.1.15]{AS00}}]\label{lem:chernoffpoisson}
	Let $X$ have a Poisson distribution with mean $\mu$. Then for any $\epsilon > 0$,
	\begin{align*}
		\Pro{X \leq (1 - \epsilon) \mu } &\leq e^{-\epsilon^2 \mu/2} ,\\
		\Pro{X \geq (1 + \epsilon) \mu }&\leq \left[ e^{\epsilon} (1+\epsilon)^{-(1+\epsilon)}  \right]^{\mu} \leq e^{-\min(\epsilon,\epsilon^2) \cdot \mu / 3}.
	\end{align*}
\end{lem}

\begin{thm}[{Optional Stopping Theorem~\cite[Corollary 17.7]{levin2009markov}}] \label{thm:ost}
	Let $(M_t)$ be a martingale and $\tau$ a stopping time. If $\Pro{\tau < \infty}$ and $|M_{t \land \tau}| \leq K$ for all $t$ and some constant $K$ where $t \land \tau := \min\{t, \tau \}$, then $\Ex{M_{\tau}} = \Ex{M_0}$.
\end{thm}

\begin{thm}[Hoeffding's Inequality \cite{hoeffding1963probability}]\label{thm:hoeffding} 
	Consider a collection of independent random variables
	$X_i \in [a_i, b_i]$ with $i \in [n]$. Then for any number $\delta > 0$,
	\begin{equation*}
		\mathbb{P}\Bigg[\Bigg|\sum\limits_{i=1}^n X_i - 
		\mathbb{E}\left[\sum\limits_{i=1}^n X_i\right]
		\Bigg| \geqslant \delta \Bigg] \leqslant 2 \cdot \exp \Bigg(\frac{-2\delta^2}
		{\sum\nolimits_{i=1}^n (b_i - a_i)^2}\Bigg).
	\end{equation*} 
\end{thm}

\bibliographystyle{plainurl}
\bibliography{avgcase}
\pagestyle{myheadings}
\markboth{}{}

\newpage


\end{document}